\documentclass[letterpaper,11pt]{article}
	\usepackage{amsfonts, amsmath, amssymb, amsthm}
	\usepackage{bbm}
	\usepackage{thmtools}
	\usepackage{algpseudocode}
	\usepackage{graphicx}
	\usepackage{stmaryrd}
	\usepackage[show]{ed}
	\usepackage[small]{complexity}
	\usepackage{ulem}
	\usepackage{xspace}
	\usepackage{lmodern}
	\usepackage[margin=1in]{geometry}
	\usepackage{microtype}
	\usepackage[pdftex,pagebackref=true,colorlinks]{hyperref}
	\usepackage{algorithm}

	\numberwithin{equation}{section}

	\theoremstyle{plain}
	\declaretheorem[numberlike=equation]{theorem}
	\declaretheorem[unnumbered,name=Theorem]{theorem*}
	\declaretheorem[numberlike=equation]{lemma}
	\declaretheorem[unnumbered,name=Lemma]{lemma*}
	\declaretheorem[numberlike=equation]{corollary}
	\declaretheorem[unnumbered,name=Corollary]{corollary*}
	
	\declaretheorem[unnumbered,name=Proposition]{proposition*}
	
	\declaretheorem[unnumbered,name=Claim]{claim*}
	
	\declaretheorem[unnumbered,name=Conjecture]{conjecture*}
	\declaretheorem[numberlike=equation]{construction}
	\declaretheorem[unnumbered,name=Construction]{construction*}
	\declaretheorem[numberlike=equation]{question}
	\declaretheorem[unnumbered,name=Question]{question*}

	\declaretheorem[numberlike=equation]{definition}
	\declaretheorem[unnumbered,name=Definition]{definition*}
	
	\declaretheorem[unnumbered,name=Example]{example*}
	
	\declaretheorem[unnumbered,name=Notation]{notation*}

	\theoremstyle{remark}
	
	\declaretheorem[unnumbered,name=Remark]{remark*}

	\newcommand{\zr}[1]{{\llbracket {#1}\rrbracket}}

	\newcommand{\la}{\langle}
	\newcommand{\ra}{\rangle}

	\DeclareMathOperator{\spn}{span}
	\DeclareMathOperator{\chara}{char}
	\DeclareMathOperator{\tr}{trace}
	\DeclareMathOperator{\Tr}{trace}
	
	\DeclareMathOperator{\supp}{Supp}
	\DeclareMathOperator{\coeffo}{\mathfrak{C}}
	\DeclareMathOperator{\GL}{GL}

	\newcommand{\coeff}[1]{\coeffo_{#1}}

	\newcommand{\F}{\mathbb{F}}

	\newcommand{\Z}{\mathbb{Z}}
	\newcommand{\N}{\mathbb{N}}
	\newcommand{\Q}{\mathbb{Q}}
	\renewcommand{\C}{\mathbb{C}}

	\newcommand{\cC}{\mathcal{C}}
	
	\newcommand{\cH}{\mathcal{H}}

	\newcommand{\cP}{\mathcal{P}}
	\newcommand{\cT}{\mathcal{T}}
	\newcommand{\Id}{\mathrm{I}}

	\newcommand{\eqdef}{:=}
	\renewcommand{\O}{\mathcal{O}}

	\newcommand{\ignore}[1]{}

\begin{document}

\title{Explicit Noether Normalization for Simultaneous Conjugation via Polynomial Identity Testing}
\author{%
Michael A.\ Forbes\thanks{Email: \texttt{miforbes@mit.edu},
Department of Electrical Engineering and Computer Science, MIT CSAIL, 32 Vassar St.,
Cambridge, MA 02139, Supported by NSF grant CCF-0939370.}
		\and
Amir Shpilka\thanks{Faculty of Computer Science, Technion --- Israel Institute of Technology, Haifa, Israel,
\texttt{shpilka@cs.technion.ac.il}.  The research leading to these results has received funding
from the European Community's Seventh Framework Programme (FP7/2007-2013) under grant agreement number 257575.}}
\maketitle

\begin{abstract}
	Mulmuley~\cite{Mulmuley12} recently gave an explicit version of Noether's Normalization lemma for ring of invariants of matrices under simultaneous conjugation, under the conjecture that there are deterministic black-box algorithms for polynomial identity testing (PIT).  He argued that this gives evidence that constructing such algorithms for PIT is beyond current techniques.  In this work, we show this is not the case.  That is, we improve Mulmuley's reduction and correspondingly weaken the conjecture regarding PIT needed to give explicit Noether Normalization.  We then observe that the weaker conjecture has recently been nearly settled by the authors (\cite{ForbesShpilka12a}), who gave quasipolynomial size hitting sets for the class of read-once oblivious algebraic branching programs (ROABPs).  This gives the desired explicit Noether Normalization unconditionally, up to quasipolynomial factors.

	As a consequence of our proof we give a deterministic parallel polynomial-time algorithm for deciding if two matrix tuples have intersecting orbit closures, under simultaneous conjugation.  

	We also study the strength of conjectures that Mulmuley requires to obtain similar results as ours.  We prove that his conjectures are stronger, in the sense that the computational model he needs PIT algorithms for is equivalent to the well-known algebraic branching program (ABP) model, which is provably stronger than the ROABP model.


	Finally, we consider the depth-3 diagonal circuit model as defined by Saxena~\cite{Saxena08}, as PIT algorithms for this model also have implications in Mulmuley's work.  Previous work (such as \cite{AgrawalSS12} and \cite{ForbesShpilka12a}) have given quasipolynomial size hitting sets for this model.  In this work, we give a much simpler construction of such hitting sets, using techniques of Shpilka and Volkovich~\cite{ShpilkaVolkovich09}.
\end{abstract}








\section{Introduction}

Many results in mathematics are non-constructive, in the sense that they establish that certain mathematical objects exist, but do not give an efficient or explicit construction of such objects, and often further work is needed to find constructive arguments.  Motivated by the recent results of Mulmuley~\cite{Mulmuley12} (henceforth ``Mulmuley'', but theorem and page numbering will refer to the full version~\cite{Mulmuley12full}), this paper studies constructive versions of the Noether Normalization Lemma from commutative algebra.  The lemma, as used in this paper, can be viewed as taking a commutative ring $R$, and finding a smaller subring $S\subseteq R$ such that $S$ captures many of the interesting properties of $R$ (see Section~\ref{sec:NNL} for a formal discussion).  Like many arguments in commutative algebra, the usual proof of this lemma does not focus on the computational considerations of how to find (a concise representation of) the desired subring $S$.  However, the area of computational commutative algebra (eg, \cite{DK,CLO}) offers methods showing how many classic results can be made constructive, and in certain cases, the algorithms are even efficient.

While constructive methods for Noether Normalization are known using Gr\"{o}bner bases (cf. \cite{CLO}), the Gr\"{o}bner basis algorithms are not efficient in the worst-case (as show by Mayr and Meyer~\cite{MayrMeyer}), and are not known to be more efficient for the problems we consider.  Diverging from the Gr\"{o}bner basis idea, Mulmuley recently observed that a constructive version of Noether Normalization is really a problem in \textit{derandomization}.  That is, given the ring $R$, if we take a sufficiently large set of ``random'' elements from $R$, then these elements will generate the desired subring $S$ of $R$.  Indeed, the usual proof of Noether Normalization makes this precise, with the appropriate algebraic meaning of ``random''. This view suggests that random sampling from $R$ is sufficient to construct $S$, and this sampling will be efficient if $R$ itself is explicitly given. While this process uses lots of randomness, the results of the derandomization literature in theoretical computer science (eg, \cite{IW97,IKW02,KabanetsImpagliazzo04}) give strong conjectural evidence that randomness in efficient algorithms is not necessary.  Applied to the problems here, there is thus strong conjectural evidence that the generators for the subring $S$ can be constructed efficiently, implying that the Noether Normalization Lemma can be made constructive.

Motivated by this connection, Mulmuley explored what are the minimal derandomization conjectures  necessary to imply an explicit form of Noether Normalization.  The existing conjectures come in two flavors.  Most derandomization hypotheses concern boolean computation, and as such are not well-suited for algebraic problems (for example, a single real number can encode infinitely many bits), but Mulmuley does give some connections in this regime. Other derandomization hypotheses directly concern algebraic computation, and using them Mulmuley gives an explicit Noether Normalization Lemma, for some explicit rings of particular interest.  In particular, Mulmuley proves that it would suffice to derandomize the \textit{polynomial identity testing} (PIT) problem in certain models, in order to obtain a derandomization of the Noether Normalization Lemma.  Mulmuley actually views this connection as an evidence that derandomizing PIT for these models is a difficult computational task (Mulmuley, p.\ 3) and calls this the {\em GCT chasm}. Although Mulmuley conjectures that it could be crossed he strongly argues that this cannot be achieved with current techniques (Mulmuley, p.\ 3):

\ignore{

\begin{quote}
	These and related results may explain in a unified way why proving lower bounds or derandomization results for arithmetic circuits in characteristic zero of even depth three or constant-depth Boolean threshold circuits, or proving uniform Boolean conjectures such as $\EXP^\NP\not\subseteq\Sigma_2^P\cap\Pi_2^P$  has turned out to be so hard, and also why the classification problems of invariant theory and algebraic geometry have turned out to be so hard from the complexity-theoretic perspective.
\end{quote}

Mulmuley calls this the {\em GCT chasm} and conjectures that it could be crossed (Mulmuley, p.\ 12):

\begin{quote}
	We conjecture that (separating) e.s.o.p's exist for the $K[V]^G$'s under consideration and the coordinate rings of explicit varieties in general: i.e., Noether's Normalization Lemma for these rings can be derandomized (in a strong form), as suggested by Theorems 1.1, 1.2, 1.3, 1.5 and 1.6, cf. Conjecture 11.1, and the GCT chasm can be crossed.
\end{quote}

However, Mulmuley has strongly suggested that crossing this chasm is difficult (Mulmuley, p.\ 3):

}

\begin{quote}
	On the negative side, the results in this article say that black-box derandomization of PIT in characteristic zero would necessarily require proving, either directly or by implication, results in algebraic geometry that seem impossible to prove on the basis of the current knowledge.
\end{quote}

In this work, we obtain a derandomization of Noether's Normalization Lemma for the problems discussed in Mulmuley's Theorems 1.1 and 1.2, using existing techniques.  These results alone have been suggested by Mulmuley (in public presentation) to contain the ``impossible'' problems mentioned above. This suggests that problems cannot be assumed to be difficult just because they originated in algebraic-geometric setting, and that one has to consider the finer structure of the problem.

In addition, just as Mulmuley's techniques extend to Noether Normalization of arbitrary quivers (Mulmuley's Theorem 4.1), our results also extend to this case, but we omit the straightforward details.  However, we do not give any results for Noether Normalization of general explicit varieties as discussed in Mulmuley's Theorem 1.5, and indeed that seems difficult given that Mulmuley's Theorem 1.6 gives an equivalence between general PIT and Noether Normalization for a certain explicit variety.\footnote{The main reason we could obtain our results is that for the variety we consider, there are explicitly known generators for the ring of invariants (as given in \autoref{generatinginvariants}), and these generators are computationally very simple. For general explicit varieties, obtaining such explicit generators is an open problem, and even if found, the generators would not likely be as computational simple as the generators of \autoref{generatinginvariants}. We refer the reader to Mulmuley's Section 10 where these issues are discussed further.}

We start by briefly describing the PIT problem.  For more details, see the survey by Shpilka and Yehudayoff~\cite{SY10}.

\subsection{Polynomial Identity Testing}

The PIT problem asks to decide whether a polynomial, given by an algebraic circuit, is the zero polynomial.  An algebraic circuit is a directed acyclic graph with a single sink (or root) node, called the \textit{output}.  Internal nodes are labeled with either a $\times$- or $+$-gate, which denote multiplication and addition respectively.  The source (or leaf) nodes, are labeled with either variables $x_i$, or elements from a given field $\F$.  An algebraic circuit computes a polynomial in the ring $\F[x_1,\ldots,x_n]$ in the natural way (as there are no cycles): each internal node in the graph computes a function of its children (either $\times$ or $+$), and the circuit itself outputs the polynomial computed by the output node. Algebraic circuits give the most natural and succinct way to represent a polynomial.

Given an algebraic circuit $C$, the PIT problem is to test if the polynomial $f$ it computes is identically zero, as a polynomial in $\F[x_1,\ldots,x_n]$. Schwartz~\cite{Schwartz80} and Zippel~\cite{Zippel79} (along with the weaker result by DeMillo and Lipton \cite{DemilloL78}) showed that if $f\not\equiv 0$ is a polynomial of degree $\le d$, and $\alpha_1,\ldots,\alpha_n\in S\subseteq\F$ are chosen uniformly at random, then $f(\alpha_1,\ldots,\alpha_n)=0$ with probability $\le d/|S|$.  It follows then that we can solve PIT efficiently using randomness, as evaluating an algebraic circuit can be done efficiently.  The main question concerning the PIT problem is whether it admits an efficient \textit{deterministic} algorithm, in the sense that it runs in polynomial time in the size of the circuit $C$. Heuristically, this problem can be viewed as replacing any usage of the Schwartz-Zippel result with a deterministic set of evaluations.

One important feature of the above randomized PIT algorithm is that it only uses the circuit $C$ by evaluating it on given inputs.  Such algorithms are called \textit{black-box} algorithms, as they treat the circuit as merely a black-box that computes some low-degree polynomial (that admits some small circuit computing it).  This is in contrast to \textit{white-box} algorithms, that probe the structure of $C$ in some way.  White-box algorithms are thus less restricted, whence deriving a black-box algorithm is a stronger result.  For the purposes of this paper, instead of referring to deterministic black-box PIT algorithms, we will use the equivalent notion of a \textit{hitting set}, which is a small set $\cH\subseteq\F^n$ of evaluation points such that for any non-zero polynomial $f$ computed by a small algebraic circuit, $f$ must evaluate to non-zero on some point in $\cH$.  A standard argument (see \cite{SY10}) shows that small hitting sets \textit{exist}, the main question is whether small \textit{explicit} hitting sets, which we now define, exist.  As usual, the notion of explicit must be defined with respect to an infinite family of objects, one object for each value of $n$.  For clarity, we abuse notation and do not discuss families, with the understanding that any objects we design will belong to an unspecified family, and that there is a single (uniform) algorithm to construct these objects that takes as input the relevant parameters.

\begin{definition}
	Let $\cC\subseteq\F[x_1,\ldots,x_n]$ be a class of polynomials.  A set $\cH\subseteq \F^n$ is a \textbf{hitting set for $\cC$} if for all $f\in\cC$, $f\equiv 0$ iff $f|_{\cH}\equiv 0$. The hitting set $\cH$ is \textbf{$t(n)$-explicit} if there is an algorithm such that given an index into $\cH$, the corresponding element of $\cH$ can be computed in $t(n)$-time, assuming unit cost arithmetic in $\F$.
\end{definition}

That is, we mean that the algorithm can perform field operations (add, subtract, multiply, divide, zero test) in $\F$ in unit time, and can start with the constants $0$ and $1$.  We will also assume the algorithm has access to an arbitrary enumeration of $\F$.  In particular, when $\F$ has characteristic 0, without loss of generality the algorithm will only produce rational numbers.

\subsection{Noether Normalization for the Invariants of Simultaneous Conjugation}\label{sec:NNL}

Mulmuley showed that when $R$ is a particular ring, then the problem of finding the subring $S$ given by Noether Normalization can be reduced to the black-box PIT problem, so that explicit hitting sets (of small size) would imply a constructive version of Noether Normalization for this ring.  The ring considered here and in Mulmuley's Theorems 1.1 and 1.2 is the ring of invariants of matrices, under the action of simultaneous conjugation.

\begin{definition}
	Let $\vec{M}$ denote a vector of $r$ matrices, each\footnote{In this work we will most often index vectors and matrices starting at zero, and will indicate this by the use of $\zr{n}$, which denotes the set $\{0,\ldots,n-1\}$.  Also, $[n]$ will be used to denote the set $\{1,\ldots,n\}$.} $\zr{n}\times\zr{n}$, whose entries are distinct variables.  Consider the action of $\GL_n(\F)$ by simultaneous conjugation on $\vec{M}$, that is, \[(M_1,\ldots,M_r)\mapsto (PM_1P^{-1},\ldots,PM_rP^{-1})\;.\]
	Define $\F[\vec{M}]^{\GL_n(\F)}$ to be the subring of $\F[\vec{M}]$ consisting of polynomials in the entries of $\vec{M}$ that are invariant under the action of $\GL_n(\F)$.  That is, \[\F[\vec{M}]^{\GL_n(\F)}\eqdef \{f|f(\vec{M})=f(P\vec{M}P^{-1}), \forall P\in \GL_n(\F)\}\;.\]
\end{definition}

Note that $\F[\vec{M}]^{\GL_n(\F)}$ is in fact a ring.  When $\F$ has characteristic zero, the following result gives an explicit set of generators for the ring of invariants.  When $\F$ has positive characteristic, the result is known not to hold (see \cite[\defaultS 2.5]{kraft2000classical}) so we will only discuss characteristic zero fields.

\begin{theorem}[\cite{Procesi76,Razmyslov74,Formanek86}]
	\label{generatinginvariants}
	Let $\F$ be a field of characteristic zero. Let $\vec{M}$ denote a vector of $r$ matrices, each $\zr{n}\times\zr{n}$, whose entries are distinct variables.  The ring $\F[\vec{M}]^{\GL_n(\F)}$ of invariants is generated by the invariants
	$\cT\eqdef\{\tr(M_{i_1}\cdots M_{i_\ell})| \vec{i}\in\zr{r}^\ell,\ell\in[n^2]\}$.
	
	Further, the ring $\F[\vec{M}]^{\GL_n(\F)}$ is not generated by the invariants
	$\{\tr(M_{i_1}\cdots M_{i_\ell})| \vec{i}\in\zr{r}^\ell,\ell\in[\lceil n^2/8\rceil]\}$.
\end{theorem}

That is, every invariant can be represented as a (multivariate) polynomial, with coefficients in $\F$, in the above generating set. Note that the above generating set is indeed a set of invariants, because the trace is cyclic, so the action of simultaneous conjugation by $P$ cancels out.

The above result is explicit in two senses.  The first sense is that all involved field constants can be efficiently constructed.  The second is that for any $f\in \cT$ and $\vec{A}$, $f(\vec{A})$ can be computed quickly.  In particular, any $f\in \cT$ can be computed by a $\poly(n,r)$-sized algebraic circuit, as matrix multiplication and trace can be computed efficiently by circuits.  We encapsulate these notions in the following definition.

\begin{definition}
	A set $\cP\subseteq\F[x_1,\ldots,x_n]$ of polynomials has \textbf{$t(n)$-explicit\footnote{It is important to contrast this with the vague notion of being \textit{mathematically} explicit.  For example, a non-trivial $n$-th root of unity is mathematically explicit, but is not computationally explicit from the perspective of the rational numbers.  Conversely, the lexicographically first function $f:\{0,1\}^{\lfloor \log\log n\rfloor}\to\{0,1\}$ with the maximum possible circuit complexity among all functions on $\lfloor \log\log n\rfloor$ bits, is computationally explicit, but arguably not mathematically explicit. While will we exclusively discuss the notion of computational explicitness in this paper, the constructions are all mathematically explicit, including the Forbes-Shpilka~\cite{ForbesShpilka12a} result cited as \autoref{thm:FS}.} $\cC$-circuits}, if there is an algorithm such that given an index into $\cP$, a circuit $C\in\cC$ can be computed in $t(n)$-time, assuming unit cost arithmetic in $\F$, such that $C$ computes the indexed $f\in\cP$.
\end{definition}

In particular, the above definition implies that the resulting circuits $C$ have size at most $t(n)$. The class of circuits $\cC$ can be the class of all algebraic circuits, or some restricted notion, such as algebraic branching programs, which are defined later in this paper. Thus, in the language of the above definition, the set of generators $\cT$ has $\poly(n,r)$-explicit algebraic circuits.

However, the above result is unsatisfactory in that the set of generators $\cT$ has size $\exp(\poly(n,r))$, which is unwieldy from a computational perspective. One could hope to find a smaller set of generators, but the lower bound in the above theorem seems a barrier in that direction. The number of generators is relevant here, as we will consider three computational problems where these generators are useful, but because of their number the resulting algorithms will be exponential-time, where one could hope for something faster. To define these problems, we first give the following standard definition from commutative algebra.

\begin{definition}
	Let $R$ be a commutative ring, and $S$ a subring.  Then $R$ is \textbf{integral} over $S$ if every element in $R$ satisfies some monic polynomial with coefficients in $S$.
\end{definition}

As an example, the algebraic closure of $\Q$ (the algebraic numbers) is integral over $\Q$. In this work the rings $R$ and $S$ will be rings of polynomials, and it is not hard to see that all polynomials in $R$ vanish at a point iff all polynomials in $S$ vanish at that point.  This can be quite useful, especially if $S$ has a small list of generators.  The statement of Noether Normalization is exactly that of providing such an $S$ with a small list of generators. The question we consider here is how to find an explicit such $S$ for the ring of invariants under simultaneous conjugation, where $S$ should be given by its generators.

\begin{question}
	\label{q:integral subring}
	Let $\F$ be an algebraically closed field of characteristic zero. Is there a small set of polynomials $\cT'\subseteq\F[\vec{M}]^{\GL_n(\F)}$ with explicit algebraic circuits, such that $\F[\vec{M}]^{\GL_n(\F)}$ is integral over the subring $S$ generated by $\cT'$?
\end{question}

We will in fact mostly concern ourselves with the next problem, which has implications for the first, when $\F$ is algebraically closed.  We first give the following definition, following Derksen and Kemper~\cite{DK}.

\begin{definition}
	A subset $\cT'\subseteq \F[\vec{M}]^{\GL_n(\F)}$ is a set of \textbf{separating invariants} if for all $\vec{A},\vec{B}\in(\F^{\zr{n}\times \zr{n}})^{\zr{r}}$, there exists an $f\in\F[\vec{M}]^{\GL_n(\F)}$ such that $f(\vec{A})\ne f(\vec{B})$ iff there exists an $f'\in \cT'$ such that $f'(\vec{A})\ne f'(\vec{B})$.
\end{definition}

As before, we will ask whether we can find an explicit construction.

\begin{question}
	\label{q:separating invariants}
	Let $\F$ have characteristic zero. Is there a small set of separating invariants $\cT'\subseteq\F[\vec{M}]^{\GL_n(\F)}$ with explicit algebraic circuits?
\end{question}

Mulmuley used the tools of geometric invariant theory~\cite{MFK}, as done in Derksen and Kemper~\cite{DK}, to note that, over algebraically closed fields, any set $\cT'$ of separating invariants will also generate a subring that $(\F^{\zr{n}\times \zr{n}})^{\zr{r}}$ is integral over.  Thus, any positive answer to \autoref{q:separating invariants} will give a positive answer to \autoref{q:integral subring}. Hence, we will focus on constructing explicit separating invariants (over any field of characteristic zero).

Note that relaxations of \autoref{q:separating invariants} can be answered positively. If we only insist on explicit separating invariants (relaxing the insistence on having few invariants), then the exponentially-large set of generators $\cT$ given in \autoref{generatinginvariants} suffices as these polynomials have small circuits and as they generate the ring of invariants, they have the required separation property. In contrast, if we only insist of a small set of separating invariants (relaxing the explicitness), then  Noether Normalization essentially shows that a \textit{non-explicit} set of separating invariants $\cT'$ of size $\poly(n,r)$ exists, basically by taking a random $\cT'$. More constructively, Mulmuley observed that Gr\"{o}bner basis techniques can construct a small set of separating invariants $\cT'$, but this set is still not explicit as such algorithms take exponential-space, so are far from efficient. In the particular case of $\F[\vec{M}]^{\GL_n(\F)}$, Mulmuley showed that the construction can occur in $\PSPACE$ unconditionally, or even $\PH$, assuming the Generalized Riemann Hypothesis. Thus, while there are explicit sets of separating invariants, and there are small sets of separating invariants, existing results do not achieve these two properties simultaneously.

The third problem is more geometric, as opposed to algebraic.  Given a tuple of matrices $\vec{A}$, we can consider the orbit of $\vec{A}$ under simultaneous conjugation as a subset of $(\F^{\zr{n}\times \zr{n}})^{\zr{r}}$. A natural computational question is to decide whether the orbits of $\vec{A}$ and $\vec{B}$ intersect.  However, from the perspective of algebraic geometry it is more natural to ask of the \textit{orbit closures} intersect. That is, we now consider $\vec{A}$ and $\vec{B}$ as lying in $(\overline{\F}^{\zr{n}\times \zr{n}})^{\zr{r}}$, where $\overline{\F}$ is the algebraic closure of $\F$.  Then, we consider the orbit closures of $\vec{A}$ and $\vec{B}$ in this larger space, where this refers to taking the orbits in $(\overline{\F}^{\zr{n}\times \zr{n}})^{\zr{r}}$ and closing them with respect to the Zariski topology. This yields the following question.

\begin{question}
	\label{q:orbit intersection}
	Let $\F$ be a field of characteristic zero.  Is there an efficient deterministic algorithm (in the unit cost arithmetic model) that, given $\vec{A},\vec{B}\in(\F^{\zr{n}\times \zr{n}})^{\zr{r}}$, decides whether the orbit closures of $\vec{A}$ and $\vec{B}$ under simultaneous conjugation have an empty intersection?
\end{question}

Mulmuley observed that by the dictionary of geometric invariant theory~\cite{MFK}, $\vec{A}$ and $\vec{B}$ have a non-empty intersection of their orbit closures iff they are not distinguishable by any set of separating invariants.  Thus, any explicit set $\cT'$ of separating invariants, would answer this question, as one could test if $f$ agrees on $\vec{A}$ and $\vec{B}$ (as $f$ is easy to compute, as it has a small circuit), for all $f\in T'$. Thus, as before, \autoref{q:orbit intersection} can be solved positively by a positive answer to \autoref{q:separating invariants}.

The main results of this paper provide positive answers to Questions~\ref{q:integral subring}, \ref{q:separating invariants} and \ref{q:orbit intersection}.

\subsection{Mulmuley's results}

Having introduced the above questions, we now summarize Mulmuley's results that show that these questions can be solved positively if one assumes that there exist explicit hitting sets for a certain subclass of algebraic circuits, which we now define.  We note that Mulmuley defines this model using linear, and not affine functions.  However, we define the model using affine functions as this allows the model to compute any polynomial (and not just homogeneous polynomials), potentially with large size.  However, this is without loss of generality, as derandomizing PIT is equally hard in the linear and the affine case, via standard homogenization techniques, see \autoref{linervsaffine}.

\begin{definition}[Mulmuley]\label{def: ST}
	A polynomial $f(x_1,\ldots,x_n)$ is computable by a width $w$, depth $d$, \textbf{trace of a matrix power} if there exists a matrix $A(x_1,\ldots,x_n)\in\F[x_1,\ldots,x_n]^{\zr{w}\times\zr{w}}$ whose entries are affine functions in $\vec{x}$  such that $f(\vec{x})=\Tr(A(\vec{x})^d)$. The \textbf{size}\footnote{One could consider the size to be $nw\log d$ because of the repeated squaring algorithm. However, in this paper our size measure is more natural as all such $d$ will be small.} of a trace of a matrix power is $nwd$.
\end{definition}

As matrix multiplication and trace both have small algebraic circuits, it follows that traces of matrix powers have small circuits.  Further, as a restricted class of algebraic circuits, we can seek to deterministically solve the PIT problem for them, and the hypothesis that this is possible is potentially weaker than the corresponding hypothesis for general algebraic circuits.  However, this hypothesis, in its black-box form, is strong enough for Mulmuley to derive implications for the above questions.

\begin{theorem*}[Part of Mulmuley's Theorems 1.1 and 3.6]
	Let $\F$ be a field of characteristic zero. Assume that there is a $t(m)$-explicit hitting set, of size $s(m)$, for traces of matrix power over $\F$ of size $m$.  Then there is a set $\cT'$ of separating invariants, of size $\poly(s(\poly(n,r)))$, with $\poly(t(\poly(n,r)))$-explicit traces of matrix powers.  Further, for algebraically closed $\F$, $\F[\vec{M}]^{\GL_n(\F)}$ is integral over the ring generated by $\cT'$.
\end{theorem*}

We briefly summarize the proof idea.  It is clear from the definitions that the generating set $\cT$ for $\F[\vec{M}]^{\GL_n(\F)}$ is a set of separating invariants, albeit larger than desired.  The proof will recombine these invariants into a smaller set, by taking suitable linear combinations.  Specifically, suppose $\vec{A}$ and $\vec{B}$ are separable, and thus $\vec{\cT}(\vec{A})\ne\vec{\cT}(\vec{B})$, where $\vec{\cT}(\vec{A})$ denotes the sequence of evaluations $(f(\vec{A}))_{f\in\cT}$.  Standard arguments about inner-products show then that $\la\vec{\cT}(\vec{A}),\vec{\alpha}\ra\ne \la\vec{\cT}(\vec{B}),\vec{\alpha}\ra$ for random values of $\vec{\alpha}$.  As a linear combination of invariants is also an invariant, it follows that $f_{\vec{\alpha}}(\vec{M})=\la\vec{\cT}(\vec{M}),\vec{\alpha}\ra$ is an invariant, and will separate $\vec{A}$ and $\vec{B}$ for random $\vec{\alpha}$.  Intuitively, one can non-explicitly derandomize this to yield a set of separating invariants by taking sufficiently many choices of $\vec{\alpha}$ and union bounding over all $\vec{A}$ and $\vec{B}$.

Note that finding such $\vec{\alpha}$ is equivalent to asking for a hitting set for the class of polynomials $\{f_{\vec{x}}(\vec{A})-f_{\vec{x}}(\vec{B}):\vec{A},\vec{B}\in(\F^{\zr{n}\times\zr{n}})^{\zr{r}}\}$, so explicitly derandomizing PIT would give an explicit set of separating invariants, as was desired.  However, as is, the above reduction is unsatisfactory in two ways.  Primarily, the resulting set of separating invariants would still be exponentially large, as one cannot, by a counting argument, construct small hitting sets for the above class of polynomials unless one can exploit structure in vectors $\vec{\cT}(\vec{A})$.  Second, the resulting invariants $f_{\vec{\alpha}}(\vec{M})$ will not have small circuits, unless, as before, one can exploit the structure of $\vec{\cT}(\vec{A})$, but now using the structure to compute the exponentially-large sum $\la\vec{\cT}(\vec{A}),\vec{\alpha}\ra$ in sub-exponential time.  Both of these problems can be overcome by showing that indeed the vector $\vec{\cT}(\vec{A})$ does have structure, in particular that it can be encoded into the coefficients of a small circuit.  The circuit class that Mulmuley uses is the trace of matrix powers model.

Assuming various plausible conjectures and using the requisite results in derandomization literature, Mulmuley showed that small explicit hitting sets exist, removing the need to outright conjecture the existence of such hitting sets.  This thus established the conjectural existence of small explicit sets of separating invariants by the above theorem.  We list one such conditional result here, noting that all such conditional results Mulmuley derived gave sets of separating invariants of quasi-polynomial size, or worse.

\begin{theorem*}[Part of Mulmuley's Theorems 1.2 and 5.1]
	Let $\F$ be a field of characteristic zero. Suppose there is a multilinear polynomial (family) $f(x_1,\ldots,x_n)\in\Z[x_1,\ldots,x_n]$ with coefficients containing at most $\poly(n)$ bits, such that $f(\vec{x})$ can be computed in $\exp(n)$-time, but $f$ cannot be computed by an algebraic circuit of size $\O(2^{\epsilon n})$ and depth $\O(n^{\epsilon})$ for some $\epsilon>0$. Then  $\F[\vec{M}]^{\GL_n(\F)}$ has a $\poly(n,r)^{\polylog(n,r)}$-size set of separating invariants, with $\poly(n,r)^{\polylog(n,r)}$-explicit traces of matrix powers.
\end{theorem*}

While the above result is conditional, unconditional results can also be derived, if randomness is allowed.  That is, by exploiting the connection between separating invariants and closed orbit intersections mentioned above, and using that PIT can be solved using randomness, Mulmuley obtains the following randomized algorithm.

\begin{theorem*}[Mulmuley's Theorem 3.8]
	Let $\F$ be a field of characteristic zero. There is an algorithm, running in randomized $\polylog(n,r)$-time using $\poly(n,r)$-processors ($\RNC$), in the unit cost arithmetic model, such that given $\vec{A},\vec{B}\in(\F^{\zr{n}\times \zr{n}})^{\zr{r}}$, one can decide whether the orbit closures of $\vec{A}$ and $\vec{B}$ under simultaneous conjugation have an empty intersection.
\end{theorem*}

Using the just mentioned conjectures, Mulmuley can also partially derandomize the above algorithm, but not to within polynomial time.

\subsection{Our Results}

We study further the connection raised by Mulmuley regarding the construction of separating invariants and the black-box PIT problem.  In particular, we more carefully study the classes of algebraic circuits arising in the reduction from Noether Normalization to PIT.  Two models are particularly important, and we define them now.

\begin{definition}[Nisan~\cite{Nisan91}]\label{def: ABP}
	A \textbf{algebraic branching program with unrestricted weights} of depth $d$ and width $\le w$, on the variables $x_1,\ldots,x_n$, is a directed acyclic graph such that
	\begin{itemize}
		\item The vertices are partitioned in $d+1$ layers $V_0,\ldots,V_d$, so that $V_0=\{s\}$ ($s$ is the source node), and $V_d=\{t\}$ ($t$ is the sink node). Further, each edge goes from $V_{i-1}$ to $V_{i}$ for some $0< i\le d$.
		\item $\max|V_i|\le w$.
		\item Each edge $e$ is weighted with a polynomial $f_e\in\F[\vec{x}]$.
	\end{itemize}

	Each $s$-$t$ path is said to compute the polynomial which is the product of the labels of its edges, and the algebraic branching program itself computes the sum over all $s$-$t$ paths of such polynomials.

	\begin{itemize}
		\item In an \textbf{algebraic branching program (ABP)}, for each edge $e$ the weight $f_e(\vec{x})$ is an affine function.  The \textbf{size} is $nwd$.
		\item In a \textbf{read-once oblivious ABP (ROABP)} of (individual) degree $<r$, we have $n\eqdef d$, and for each edge $e$ from $V_{i-1}$ to $V_i$, the weight is a univariate polynomial $f_e(x_i)\in\F[x_i]$ of degree $<r$. The \textbf{size} is $dwr$.
	\end{itemize}
\end{definition}

In the definition of ROABPs we will exclusively focus on individual degree, and thus will use the term ``degree'' (in \autoref{sec:diagonal} we will use the more usual total degree, for a different class of circuits). The ROABP model is called \textit{oblivious} because the variable order $x_1<\cdots<x_d$ is fixed. The model is called \textit{read-once} because the variables are only accessed on one layer in the graph.

The ABP model is a standard algebraic model that is at least as powerful as algebraic formulas, as shown by Valiant~\cite{Valiant79a}, and can be simulated by algebraic circuits.  As shown by Berkowitz~\cite{Berkowitz84}, the determinant can be computed by a small ABP over any field.  See Shpilka and Yehuydayoff~\cite{SY10} for more on this model.

The ROABP model arose in prior work of the authors (\cite{ForbesShpilka12a}) as a natural model of algebraic computation capturing several other existing models.  This model can also be seen as an algebraic analogue of the boolean model of computation known as the read-once oblivious branching program model, which is a non-uniform analogue of the complexity class $\RL$.  See Forbes and Shpilka~\cite{ForbesShpilka12a} for more of a discussion on the motivation of this class.

Note that a polynomial computed by an ROABP of size $s$ can be computed by an ABP of size $\poly(s)$.  The converse is not true, as Nisan~\cite{Nisan91} gave exponential lower bounds for the size of non-commutative ABPs computing the determinant, and non-commutative ABPs encompass ROABPs, while as mentioned above Berkowitz~\cite{Berkowitz84} showed the determinant can be computed by small ABPs.  Thus the ROABP model is strictly weaker in computational power than the ABP model.

While there are no efficient (white-box or black-box) PIT algorithms for ABPs, we established in prior work (\cite{ForbesShpilka12a}) a quasi-polynomial sized hitting set for ROABPs.  This hitting set will be at the heart of our main result.

\begin{theorem}[\cite{ForbesShpilka12a}]\label{thm:FS}
	Let $\cC$ be the set of $d$-variate polynomials computable by depth $d$, width $\le w$, degree $<r$ ROABPs.  If $|\F|\ge \poly(d,w,r)$, then $\mathcal{C}$ has a $\poly(d,w,r)$-explicit hitting set $\cH\subseteq\F^d$, of size $\le \poly(d,w,r)^{\O(\lg d)}$.  Further, if $\F$ has characteristic zero then $\cH\subseteq\Q^d$.
\end{theorem}

Our contributions are split into four sections.

\paragraph{The computational power of traces of matrix powers:} We study the model of algebraic computation, traces of matrix powers, shown by Mulmuley to have implications for derandomizing Noether Normalization.  In particular, as this model is a restricted class of algebraic circuits, we can ask: how restricted is it?  If this model was particularly simple, it would suggest that derandomizing PIT for this class would be a viable approach to derandomizing Noether Normalization.  In contrast, if this model of computation is sufficiently general, then given the difficulty of derandomizing PIT for such general models, using Mulmuley's reduction to unconditionally derandomize Noether Normalization could be a formidable challenge. In this work, we show it is the latter case, proving the following theorem.

\begin{theorem*}[\autoref{thm: equivalence}]
	The computational models of algebraic branching programs and traces of matrix powers are equivalent, up to polynomial blow up in size.
\end{theorem*}

\paragraph{Derandomizing Noether Normalization via an improved reduction to PIT:} This section contains the main results of our paper. Given the above result, and the lack of progress on derandomizing PIT in such general models such as ABPs, it might seem that derandomizing Noether Normalization for simultaneous conjugation is challenging.  However, we show this is not true, by showing that derandomization of black-box PIT for ROABPs suffices for derandomizing Noether Normalization for simultaneous conjugation.  By then invoking our prior work on hitting sets for ROABPs cited as \autoref{thm:FS}, we establish the following theorems, giving quasi-affirmative answers to Questions~\ref{q:integral subring}, \ref{q:separating invariants} and \ref{q:orbit intersection}. Furthermore, our results are proved unconditionally and are at least as strong as the conditional results Mulmuley obtains by assuming strong conjectures such as the Generalized Riemann Hypothesis or strong lower bound results.

Specifically, we prove the following theorem which gives an explicit set of separating invariants (see \autoref{q:separating invariants}).

\begin{theorem}\label{thm:integral subring}
	Let $\F$ be a field of characteristic zero. There is a $\poly(n,r)^{\O(\log(n))}$-sized set $\cT_\cH$ of separating invariants, with $\poly(n,r)$-explicit ABPs. That is, $\cT_\cH\subseteq\F[\vec{M}]^{\GL_n(\F)}$, and for any $\vec{A},\vec{B}\in(\F^{\zr{n}\times \zr{n}})^{\zr{r}}$, $f(\vec{A})\ne f(\vec{B})$ for some $f\in\F[\vec{M}]^{\GL_n(\F)}$ iff $f'(\vec{A})\ne f'(\vec{B})$ for some $f'\in T_\cH$.
\end{theorem}

As a consequence of \autoref{thm:integral subring} and the discussion in \autoref{sec:NNL} we obtain the following corollary that gives a positive answer to \autoref{q:integral subring}. In particular, it provides a derandomization of Noether Normalization Lemma for the ring of invariants of simultaneous conjugation.

\begin{corollary}\label{Cor:main integral}
	Let $\F$ be an algebraically closed field of characteristic zero. Let $\cT_\cH$ be the set guaranteed by \autoref{thm:integral subring}. Then, $\F[\vec{M}]^{\GL_n(\F)}$ is integral over the subring generated by $T_\cH$.
\end{corollary}

For deciding intersection of orbit closures, \autoref{q:orbit intersection}, the natural extension of \autoref{thm:integral subring}, as argued in \autoref{sec:NNL}, would yield a quasi-polynomial-time algorithm for deciding orbit closure intersection.  However, by replacing the black-box PIT results for ROABPs of Forbes and Shpilka~\cite{ForbesShpilka12a} by the white-box PIT results by Raz and Shpilka~\cite{RazShpilka05} (as as well as follow-up work by Arvind, Joglekar and Srinivasan~\cite{ArvindJS09}), we can obtain the following better algorithm for deciding orbit closure intersection, proving a strong positive answer to \autoref{q:orbit intersection}.

\begin{theorem}\label{thm:orbit intersection}
	Let $\F$ be a field of characteristic zero. There is an algorithm, running in deterministic $\polylog(n,r)$-time using $\poly(n,r)$-processors ($\NC$), in the unit cost arithmetic model, such that given $\vec{A},\vec{B}\in(\F^{\zr{n}\times \zr{n}})^{\zr{r}}$, one can decide whether the orbit closures of $\vec{A}$ and $\vec{B}$ under simultaneous conjugation have an empty intersection.
\end{theorem}

As mentioned above, Mulmuley also gets results for Noether Normalization of arbitrary quivers (Mulmuley's Theorem 4.1) in a generalization of the results on simultaneous conjugation.  The main difference is a generalization of the list of generators given in \autoref{generatinginvariants} to arbitrary quivers, as given by Le Bruyn and Procesi~\cite{LeBruynProcesi90}. Our improved reduction to PIT, involving ROABPs instead of ABPS, also generalizes to this case, so analogous results to the above three theorems are readily attained.  However, to avoid discussing the definition of quivers, we do not list the details here.

\paragraph{PIT for Depth-3 Diagonal Circuits:} Mulmuley's Theorem 1.4 showed that Noether Normalization for representations of $\SL_m(\F)$ can be reduced, when $m$ is constant, to black-box PIT of a subclass of circuits known as \textit{depth-3 diagonal circuits}, see \autoref{sec:diagonal} for a definition.  This class of circuits (along with a depth-4 version) was introduced in Saxena~\cite{Saxena08}, who gave a polynomial-time white-box PIT algorithm, via a reduction to the white-box PIT algorithm for non-commutative ABPs of Raz and Shpilka~\cite{RazShpilka05}.  Saha, Saptharishi and Saxena~\cite{SahaSS11} (among other things) generalized these results to the depth-4 semi-diagonal model.  Agrawal, Saha and Saxena~\cite{AgrawalSS12} gave (among other things) a quasipolynomial size hitting set for this model, by showing that any such circuit can be shifted so that there is a small-support monomial, which can be found via brute-force.  In independent work, the present authors (in \cite{ForbesShpilka12a}) also established (among other things) a quasipolynomial size hitting set for this model.  This was done by showing that the depth-4 semi-diagonal model is efficiently simulated by the ROABP model.  Further, this was done in two ways: the first was an explicit reduction by using the duality ideas of Saxena~\cite{Saxena08}, and the second was to show that the diagonal model has a small space of derivatives in a certain sense, and that ROABPs can efficiently compute any polynomial with that sort of small space of derivatives.  Some aspects of this model are also present in the work of Gupta-Kamath-Kayal-Saptharishi~\cite{GuptaKKS13} showing that (arbitrary) depth-3 formulas capture, in a sense, the entire complexity of arbitrary algebraic circuits.

Here, we give a simpler proof that the depth-3 diagonal circuit model has a quasipolynomial size hitting set.  This is done using the techniques of \cite{ShpilkaVolkovich09}, and have some similarities with the work of Agrawal, Saha and Saxena~\cite{AgrawalSS12}.  In particular, we show the entire space of derivatives is small, for depth-3 model (but not the depth-4 model).  We then show that this implies such polynomials must contain a monomial of logarithmic support, which can be found via brute-force in quasipolynomial time.  Unlike the work of Agrawal, Saha and Saxena~\cite{AgrawalSS12}, no shifts are required for this small monomial to exist.  Thus, we get the following theorem.

\begin{theorem*}[\autoref{diagonalpit}]
	Let $\F$ be a field with size $\ge d+1$. Then there is a $\poly(n,d,\log(s))$-explicit hitting set of size $\poly(n,d)^{\O(\log s)}$ for the class of $n$-variate, degree $\le d$, depth-3 diagonal circuits of size $\le s$.
\end{theorem*}

\paragraph{Deciding (non-closed) orbit membership via PIT:} The results mentioned thus far have taken an algebro-geometric approach to studying the orbits of tuples of matrices under simultaneous conjugation, as they take this geometric action and study the algebraic structure of its invariants.  This perspective, by the very continuous nature of polynomials, can only access the orbit closures under this group action.  For example, working over $\C$, define
\begin{align*}
	A_\epsilon
	&\eqdef\begin{bmatrix}
		1&\epsilon\\
		0&1
	\end{bmatrix}
	&P_\delta
	&\eqdef\begin{bmatrix}
		\delta&0\\
		0&1
	\end{bmatrix}\;,
\end{align*}
and note that for any $\epsilon$ and for any $\delta\ne 0$, $P_\delta A_\epsilon P_\delta^{-1}=A_{\epsilon\delta}$.  It follows that for any polynomial $f$ invariant under simultaneous conjugation of $2\times 2$ matrices, that $f(A_\epsilon)$ is independent of $\epsilon$, as $f$ is continuous and we can take $\epsilon\to 0$.  However, for any $\epsilon\ne 0$, $A_\epsilon$ is not conjugate to $A_0=\Id_2$, the identity matrix, or equivalently, $A_\epsilon$ and $A_0$ are not in the same orbit.  Thus, from the perspective of invariants $A_\epsilon$ and $A_0$ are the same, despite being in different orbits.


One can ask the analogue of \autoref{q:orbit intersection}, but for orbits as opposed to orbit closures.  Note that by the invertibility of the group action, two orbits must either be disjoint, or equal.  Thus, we equivalently ask for an algorithm to the orbit membership problem for simultaneous conjugation.  That is, given $\vec{A},\vec{B}\in(\F^{\zr{n}\times \zr{n}})^{\zr{r}}$ is there an invertible $P\in\F^{\zr{n}\times\zr{n}}$ such that $\vec{B}=P\vec{A}P^{-1}$.  

Several interesting cases of this problem were solved: Chistov, Ivanyos and Karpinski~\cite{ChistovIK97}  gave a deterministic polynomial time algorithm over finite fields  and over algebraic number fields; Sergeichuk~\cite{Sergeichuk2000} gave\footnote{In his paper Sergeichuk gives credit for the algorithm to Belitski{\u \i}~\cite{belitskii1983}.} a deterministic algorithm over any field, that runs in polynomial time when supplied with an oracle for finding roots of polynomials.\footnote{This is not how the result is stated in \cite{Sergeichuk2000}, but this is what one needs to make the algorithm efficient. In fact, to make the algorithm of Sergeichuk run in polynomial space one needs to make another assumption that would allow writing down the eigenvalues of all matrices involved in polynomial space. For example, one such assumption would be that they all belong to the some polynomial degree extension field (Grochow~\cite{Grochow13}).} Chistov-Ivanyos-Karpinski also mentioned that a randomized polynomial-time algorithm for the problem follows from the work of Schwartz and Zippel \cite{Schwartz80,Zippel79}.
In conversations with Yekhanin~\cite{Yekhanin2013}, we also discovered this randomized algorithm, showing that this problem is reducible to PIT for ABPs. Because of its close relation to the rest of this work, we include for completeness the proof of the following theorem.

\begin{theorem*}[\autoref{orbitmembership}]
	Let $\F$ be a field of size $\ge\poly(n)$.  Then the orbit membership problem, for simultaneous conjugation, is reducible to polynomial identity testing for ABPs.  In particular, this problem has a randomized, parallel, polynomial-time algorithm.
\end{theorem*}

\subsection{Notation}\label{notation}

Given a vector of polynomials $\vec{f}\in \F[\vec{x}]^n$ and an exponent vector $\vec{i}\in\N^n$, we write $\vec{f}^{\vec{i}}$  for $f_1^{i_1}\cdots f_n^{i_n}$.

Given a polynomial $f\in\F[\vec{x}]$, we write $\coeff{\vec{x}^{\vec{i}}}(f)$ to denote the coefficient of $\vec{x}^{\vec{i}}$ in $f$.  For a matrix $M\in\F[\vec{x}]^{\zr{r}\times \zr{r}}$, we write $\coeff{\vec{x}^{\vec{i}}}(M)$ to denote the $r\times r$ $\F$-matrix, with the $\coeff{\vec{x}^{\vec{i}}}$ operator applied to each entry. When we write ``$f\in\F[\vec{x}][\vec{y}]$'', we will treat $f$ as a polynomial in the variables $\vec{y}$, whose coefficients are polynomials in the variables $\vec{x}$, and correspondingly will write $\coeff{\vec{y}^{\vec{j}}}(f)$ to extract the polynomial in $\vec{x}$ that is the coefficient of the monomial $\vec{y}^{\vec{j}}$ in $f$.

\subsection{Organization}

In \autoref{sec:ROABP} we give the necessary background on ROABPs. We prove our main results about explicit Noether Normalization in \autoref{sec:main}. 

The rest of our results appear in the following order.  We give the equivalence between the trace of matrix power and ABP models of computation in \autoref{sec:equivalent}.  We give the hitting set for depth-3 diagonal circuits in \autoref{sec:diagonal}, using Hasse derivatives as defined in \autoref{sec:hasse}.  In \autoref{sec:orbitmem} we give the reduction from the orbit membership problem to PIT.


\section{Properties of Algebraic Branching Programs}\label{sec:ROABP}

We first derive some simple properties of ABPs, as well as ROABPs, that show their tight connection with matrix products, and traces of matrix products.  We begin with the following connection between an ABP with unrestricted weights, and the product of its adjacency matrices.  As the lemma is proved in generality, it will apply to ABPs and ROABPs, and we will use it for both.

\begin{lemma}
	\label{roabpadj}
	Let $f\in\F[x_1,\ldots,x_n]$ be computed by a depth $d$, width $\le w$ ABP with unrestricted weights, such that the variable layers are $V_0,\ldots,V_d$.  For $0<i\le d$, define $M_i\in\F[\vec{x}]^{V_{i-1}\times V_i}$ such that the $(u,v)$-th entry in $M_i$ is the label on the edge from $u\in V_{i-1}$ to $v\in V_i$, or 0 if no such edge exists. Then, when treating $\F[\vec{x}]^{\zr{1}\times\zr{1}}=\F[\vec{x}]$, \[f(\vec{x})=\prod_{i\in[d]} M_i(\vec{x})\eqdef M_1(\vec{x}) M_2(\vec{x})\cdots M_d(\vec{x})\;.\]
	Further, for an ABP, the matrix $M_i$ has entries which are affine forms, and for an ROABP, the matrix $M_i$ has entries which are univariate polynomials in $x_i$ of degree $<r$.
\end{lemma}
\begin{proof}
	Expanding the matrix multiplication completely, one sees that it is a sum of the product of the labels of the $s$-$t$ paths such that the $i$-th edge goes from $V_{i-1}$ to $V_i$.  By the layered structure of the ABP, this is all such paths, so this sum computes $f(\vec{x})$.
\end{proof}

The above lemma shows that one can easily convert an ABP or ROABP into a matrix product, where the entries of matrices obey the same restrictions as the weights in the ABP or ROABP.  The above lemma gives matrices with varying sizes, and it will be more convenient to have square matrices, which can be done by padding, as shown in the next lemma.

\begin{lemma}
	\label{roabpmatrix}
	Let $f\in\F[x_1,\ldots,x_n]$ be computed by a depth $d$, width $\le w$ ABP with unrestricted weights.  Then for $i\in[d]$, there are matrices $M_i\in\F[\vec{x}]^{\zr{w}\times\zr{w}}$ such that in block notation,
	\[	
		\begin{bmatrix}
			f(\vec{x})&0\\
			0&0
		\end{bmatrix}
		=\prod_{i\in[d]}M_i(\vec{x})\;,
	\]
	that is, $\prod_{i\in[d]}M_i(\vec{x})$ contains a single non-zero entry located at $(0,0)$, which contains the polynomial $f(\vec{x})$.  Conversely, any such polynomial $f$ such that
	\[
		f(\vec{x})=\left(\prod_{i\in[d]}M_i(\vec{x})\right)_{(0,0)}
	\]
	can be computed by a depth $d$, width $w$, ABP with unrestricted weights.

	Further, for the specific cases of ABPs and ROABPs, the entries in the $M_i$ are restricted: for ABPs the matrix $M_i$ has entries which are affine forms, and for an ROABP the matrix $M_i$ has entries which are univariate polynomials in $x_i$ of degree $<r$.
\end{lemma}
\begin{proof}
	\uline{ABP$\implies$ matrices:} \autoref{roabpadj} furnishes $M_i$ such that $f(\vec{x})=\prod_i M_i(\vec{x})$.  Let $M'_i\in\F[\vec{x}]^{\zr{w}\times\zr{w}}$ be $M_i$ padded with zeroes to become $\zr{w}\times\zr{w}$ sized, that is,
	\[M'_i(\vec{x})\eqdef
		\begin{bmatrix}
			M_i&0\\
			0&0
		\end{bmatrix}
	\]
	where we use block-matrix notation.  By the properties of block-matrix multiplication, it follows that
	\[
		\prod_i M'_i=
		\prod_i
		\begin{bmatrix}
			M_i&0\\
			0&0
		\end{bmatrix}
		=
		\begin{bmatrix}
			\prod_iM_i&0\\
			0&0
		\end{bmatrix}
		=
		\begin{bmatrix}
			f(\vec{x})&0\\
			0&0
		\end{bmatrix}
	\]
	as desired. One can observe that the use of \autoref{roabpadj} implies that the $M_i$ have the desired restrictions on their entries.

	\uline{matrices$\implies$ ABP:} For $1<i<d$ define $M'_i\eqdef M_i$.  Define $M'_1$ to be the $0$-th row of $M_1$, and define $M'_d$ to be the $0$-th column of $M_d$.  Then it follows that \[f(\vec{x})=\left(\prod_i M_i(\vec{x})\right)_{(0,0)}=\prod_i M'_i(\vec{x})\]
	and that the $M'_i$ have at most $w$ rows and columns. Using the $M'_i$ as the adjacency matrices in an ABP with unrestricted weights, it follows by \autoref{roabpadj} that $f$ is computed by a depth $d$, width $w$ ABP with unrestricted weights.  Further, the use of this lemma shows the entry restrictions for ABPs and ROABPs are respected, to the result holds for these models as well.
\end{proof}

We next observe that small-size ABPs and ROABPs are respectively closed under addition and multiplication.

\begin{lemma}
	\label{addROABP}
	Let $f,f'\in\F[x_1,\ldots,x_n]$ be computed by depth $d$ ABPs with unrestricted weights, where $f$ is computed in width $\le w$ and $f'$ is computed in width $\le w'$.  Then $f+f'$ and $f-f'$ are computed by depth $d$, width $\le (w+w')$ ABPs with unrestricted weights. Further, this also holds for the ABP model, and the ROABP model with degree $<r$.
\end{lemma}
\begin{proof}
	\uline{$f+f'$:} Consider the ABPs with unrestricted weights for $f$ and $f'$.  As they have the same depth, we can align their $d+1$ layers.  Consider them together as a new ABP, by merging the two source nodes, and merging the two sink nodes.  This is an ABP and has the desired depth, width and degree.  Note that it computes $f+f'$, as any source-sink path must either go entirely through the ABP for $f$, or entirely though the ABP for $f'$, i.e. there are no cross-terms.  Thus, the sum over all paths can be decomposed into the paths for $f$ and the paths for $f'$, showing that the computation is $f+f'$ as desired.

	Note that in the ABP case, all edge weights are still affine functions, and in the ROABP case, the edge weights are still univariate polynomials of degree $<r$ for the relevant variables, as we aligned the layers between $f$ and $f'$.

	\uline{$f-f'$:} This follows by showing that if $f'$ is computable by an ABP with unrestricted weights then so is $-f'$, all within the same bounds.  To see this, observe that by flipping the sign of all edges from $V_0$ to $V_1$ in the computation for $f'$, each source-sink path in the computation for $f'$ will have its sign flipped, and thus the entire sum will be negated, computing $-f'$ as desired.

	Note that the allowed edge weights for ABPs and ROABPs are closed under linear combinations, so the above computation of $-f$ is also valid in these models.
\end{proof}


\section{Reducing Noether Normalization to Read-once Oblivious Algebraic Branching Programs}\label{sec:main}

In this section we construct a small set of explicit separating invariants for simultaneous conjugation.  We do so by constructing a single ROABP that encodes the entire generating set $\cT$ for $\F[\vec{M}]^{\GL_n(\F)}$, as given by \autoref{generatinginvariants}.  We then use hitting sets for ROABPs to efficiently extract the separating invariants from this ROABP.  We begin with the construction\footnote{There are some slightly better versions of this construction, as well as a way to more efficiently use the hitting sets of \autoref{thm:FS}.  However, these modifications make the presentation slightly less modular, and do not improve the results by more than a polynomial factor, so we do not pursue these details.} of the ROABP.

\begin{construction}
	\label{roabpinv}
	Let $n,r,\ell\ge 1$. Let $\vec{M}$ denote a vector of $r$ matrices, each $\zr{n}\times\zr{n}$, whose entries are distinct variables.  Define
	\[M(x)\eqdef \sum_{i\in\zr{r}} M_ix^i\]
	and, for the $\ell$ variables $\vec{x}$, define
	\[f_\ell(\vec{M},\vec{x})\eqdef \tr(M(x_1)\cdots M(x_\ell)).\]
\end{construction}

The following lemma shows that these polynomials $f_\ell(\vec{M},\vec{x})$ can be computed by small ROABPs, when $\vec{x}$ is variable and $\vec{M}$ is constant.

\begin{lemma}
	\label{diffofinv}
	Assume the setup of \autoref{roabpinv}. Let $\vec{A},\vec{B}\in(\F^{\zr{n}\times \zr{n}})^{\zr{r}}$. Then $f_\ell(\vec{A},\vec{x})-f_\ell(\vec{B},\vec{x})$ can be computed by a width $2n^2$, depth $\ell$, degree $<r$ ROABP.
\end{lemma}
\begin{proof}
	Observe that $f_\ell(\vec{A},\vec{x})=\tr(A(x_1)\cdots A(x_\ell))=\sum_{i\in\zr{n}} (\prod_{j=1}^\ell A(x_j))_{(i,i)}$.  By applying a permutation of indices, and appealing to \autoref{roabpmatrix}, we see that $(\prod_{j=1}^\ell A(x_j))_{(i,i)}$ is computable by a depth $\ell$, width $n$, degree $<r$ ROABP, for each $i$.  Thus, appealing to \autoref{addROABP} completes the claim.
\end{proof}

Alternatively, when $\vec{x}$ is constant, and the matrices $\vec{M}$ are variable, then $f_\ell(\vec{M},\vec{x})$ can be computed by a small ABP.

\begin{lemma}
	\label{explicitinv}
	Assume the setup of \autoref{roabpinv}. Let $\vec{\alpha}\in\F^\ell$. Then $f_\ell(\vec{M},\vec{\alpha})$ can be computed by a width $n^2$, depth $\ell$ ABP, and this ABP is constructable in $\poly(n,r,\ell)$ steps.
\end{lemma}
\begin{proof}
	Observe that $f_\ell(\vec{M},\vec{\alpha})=\tr(M(\alpha_1)\cdots M(\alpha_\ell))$, and that each $M(\alpha_i)$ is an $\zr{n}\times\zr{n}$ matrix with entries that are affine forms in the $\vec{M}$. Thus, just as in \autoref{diffofinv}, we can compute this trace in width $n^2$, depth $\ell$ ABP by appealing to \autoref{roabpmatrix} and \autoref{addROABP}.  It is straightforward to observe that the construction of this ABP runs in the desired time bounds, as the above lemmas are constructive.
\end{proof}

Our next two lemmas highlight the connection between the polynomials in \autoref{roabpinv} and the generators of the ring of invariants provided by \autoref{generatinginvariants}. Namely, they show that the generators in the set $\cT$ of \autoref{generatinginvariants} are faithfully encoded as coefficients of the polynomial $f_\ell(\vec{M},\vec{x})$, when viewing this polynomial as lying in the ring $\F[\vec{M}][\vec{x}]$. Note here that we use the $\coeff{\vec{x}^{\vec{i}}}$ notation as defined in \autoref{notation}.

\begin{lemma}
	\label{roabpcoeff}
	Assume the setup of \autoref{roabpinv}. Then for $\vec{i}\in\N^\ell$, taking coefficients in $\F[\vec{M}][\vec{x}]$,
	\begin{equation*}
		\coeff{\vec{x}^{\vec{i}}}(f_\ell(\vec{M},\vec{x}))=
			\begin{cases}
				\tr(M_{i_1}\cdots M_{i_\ell})	& \text{if } \vec{i}\in\zr{r}^{\vec{\ell}}\\
				0				& \text{else}
			\end{cases}.
	\end{equation*}
\end{lemma}
\begin{proof}
	Consider $f_\ell(\vec{M},\vec{x})$ as a polynomial in $\F[\vec{M}][\vec{x}]$. Taking coefficients, we see that
	\begin{align*}
		\coeff{\vec{x}^{\vec{i}}}(f_\ell(\vec{M},\vec{x}))
		&=\coeff{\vec{x}^{\vec{i}}}\left(\tr(M(x_1)\cdots M(x_\ell))\right)\\
		&=\coeff{\vec{x}^{\vec{i}}}\left( \tr\left( \left(\sum_{j_1\in\zr{r}}M_{j_1}x_1^{j_1}\right)\cdots\left(\sum_{j_\ell\in\zr{r}}M_{j_\ell}x_\ell^{j_\ell}\right)\right)\right)\\
		&=\coeff{\vec{x}^{\vec{i}}}\left(\tr\left(\sum_{j_1,\ldots,j_\ell\in\zr{r}} M_{j_1}\cdots M_{j_\ell} x_1^{j_1}\cdots x_\ell^{j_\ell}\right)\right)\\
		\intertext{by linearity of the trace,}
		&=\coeff{\vec{x}^{\vec{i}}}\left(\sum_{\vec{j}\in\zr{r}^\ell} \tr(M_{j_1}\cdots M_{j_\ell})\vec{x}^{\vec{j}}\right)\\
		&=
			\begin{cases}
				\tr(M_{i_1}\cdots M_{i_\ell})	& \text{if } \vec{i}\in\zr{r}^{\vec{\ell}}\\
				0				& \text{else}
			\end{cases}.
		\qedhere
	\end{align*}
\end{proof}

As the above lemma shows that $f_\ell(\vec{M},\vec{x})$ encodes all of the generators $\cT$, it follows that $\vec{A}$ and $\vec{B}$ agree on the generators $\cT$ iff they agree on $f_\ell(\vec{M},\vec{x})$.

\begin{lemma}
	\label{coeffembed}
	Assume the setup of \autoref{roabpinv}. Let $\vec{A},\vec{B}\in(\F^{\zr{n}\times\zr{n}})^{\zr{r}}$ and $\ell\ge 1$. Then $\tr(A_{i_1}\cdots A_{i_\ell})=\tr(B_{i_1}\cdots B_{i_\ell})$ for all $\vec{i}\in\zr{r}^\ell$ iff $f_\ell(\vec{A},\vec{x})=f_\ell(\vec{B},\vec{x})$, where this second equality is as polynomials in the ring $\F[\vec{x}]$.
\end{lemma}
\begin{proof}
	The two polynomials $f_\ell(\vec{A},\vec{x})$, $f_\ell(\vec{B},\vec{x})$ are equal iff all of their coefficients are equal.  By \autoref{roabpcoeff}, this is exactly the statement that $\tr(A_{i_1}\cdots A_{i_\ell})=\tr(B_{i_1}\cdots B_{i_\ell})$ for all $\vec{i}\in\zr{r}^\ell$.
\end{proof}

Hence, the polynomials $f_\ell(\vec{M},\vec{x})$ capture the generators $\cT$ of $\F[\vec{M}]^{\GL_n(\F)}$ thus in a sense capturing the entire ring $\F[\vec{M}]^{\GL_n(\F)}$ also.

\begin{corollary}
	\label{reducetoroabp}
	Assume the setup of \autoref{roabpinv}. Let $\F$ be a field of characteristic zero. Let $\vec{A},\vec{B}\in(\F^{\zr{n}\times \zr{n}})^{\zr{r}}$.  Then $f(\vec{A})= f(\vec{B})$ for all $f\in\F[\vec{M}]^{\GL_n(\F)}$ iff $f_\ell(\vec{A},\vec{x})= f_\ell(\vec{B},\vec{x})$ for all $\ell\in[n^2]$, where the second equality is as polynomials in the ring $\F[\vec{x}]$.
\end{corollary}
\begin{proof}
	By \autoref{coeffembed}, $f_\ell(\vec{A},\vec{x})=f_\ell(\vec{B},\vec{x})$ for $\ell\in[n^2]$ iff $g(\vec{A})=g(\vec{B})$ for all $g$ of the form $g(\vec{M})=\tr(M_{i_1}\cdots M_{i_\ell})$ for $\vec{i}\in\zr{r}^\ell$ and $\ell\in[n^2]$.  This set of $g$ is exactly the set $\cT$ of \autoref{generatinginvariants}, and by that result $\cT$ generates $\F[\vec{M}]^{\GL_n(\F)}$, which implies that the polynomials in $\cT$ agree on $\vec{A}$ and $\vec{B}$ iff all the polynomials in $\F[\vec{M}]^{\GL_n(\F)}$ agree on $\vec{A}$ and $\vec{B}$.  Thus $f_\ell(\vec{A},\vec{x})=f_\ell(\vec{B},\vec{x})$ for $\ell\in[n^2]$ iff $f(\vec{A})=f(\vec{B})$ for all $f\in\F[\vec{M}]^{\GL_n(\F)}$.
\end{proof}

Thus having reduced the question of whether $\F[\vec{M}]^{\GL_n(\F)}$ separates $\vec{A}$ and $\vec{B}$, to the question of whether some $f_\ell(\vec{M},\vec{x})$ separates $\vec{A}$ and $\vec{B}$, we now seek to remove the need for the indeterminates $\vec{x}$.  Specifically, we will replace them by the evaluation points of a hitting set, as shown in the next construction.
	
\begin{construction}\label{const:hittingset}
	Assume the setup of \autoref{roabpinv}. Let $\cH\subseteq\F^{n^2}$ be a $t(n,r)$-explicit hitting set for width $\le 2n^2$, depth $n^2$, degree $<r$ ROABPs. Define
	\[\cT_\cH\eqdef \{f_\ell(\vec{M},\vec{\alpha})|\vec{\alpha}\in\cH, \ell\in[n^2]\}\;,\]
	where if $\ell<n^2$ we use the first $\ell$ variables of $\alpha$ for the values of $\vec{x}$ in the substitution.
\end{construction}

We now prove the main theorems, showing how to construct small sets of explicit separating invariants.  We first do this for an arbitrary hitting set, then plug in the hitting set given in our previous work as stated in  \autoref{thm:FS}.

\begin{theorem}
	\label{thm:sepinv}
	Assume the setup of \autoref{const:hittingset}. Let $\F$ be a field of characteristic zero. Then $\cT_\cH$ is a set of size $n^2|\cH|$ of homogeneous separating invariants with $\poly(t(n,r),n,r)$-explicit ABPs.  That is, $\cT_\cH\subseteq\F[\vec{M}]^{\GL_n(\F)}$, and for any $\vec{A},\vec{B}\in(\F^{\zr{n}\times \zr{n}})^{\zr{r}}$, $f(\vec{A})\ne f(\vec{B})$ for some $f\in\F[\vec{M}]^{\GL_n(\F)}$ iff $f'(\vec{A})\ne f'(\vec{B})$ for some $f'\in \cT_\cH$, and each such $f'$ is computed by an explicit ABP.
\end{theorem}
\begin{proof}
	\uline{$\cT_\cH$ has explicit ABPs:}  We can index $\cT_\cH$ by $\ell\in[n^2]$ and an index in $\cH$.  Given an index in $\cH$ we can, by the explicitness of $\cH$, compute the associated $\vec{\alpha}\in\cH$ in $t(n,r)$ steps.  By \autoref{explicitinv} we can compute an ABP for $f_\ell(\vec{M},\vec{\alpha})$ in $\poly(n,r)$ steps, as $\ell\le n^2$, as desired.

	\uline{$f\in\cT_\cH$ are homogeneous:} This is clear by construction.

	\uline{$|T_\cH|=n^2|\cH|$:} For each $\ell\in[n^2]$, we use one invariant per point in $\cH$.

	\uline{$T_\cH\subseteq\F[\vec{M}]^{\GL_n(\F)}$:} For any $\ell$, \autoref{roabpcoeff} shows that the coefficients of $f_\ell(\vec{M},\vec{x})$ (when regarded as polynomials in $\F[\vec{M}][\vec{x}]$) are traces of products of the matrices in $\vec{M}$, so these coefficients are invariant under simultaneous conjugation of $\vec{M}$.  It follows then for any $\alpha$, $f_\ell(\vec{M},\alpha)$ is a linear combination of invariants, and thus is an invariant.  As $T_\cH$ consists of exactly such polynomials, it is contained in the ring of all invariants.

	\uline{$\F[\vec{M}]^{\GL_n(\F)}$ separates $\iff$ $T_\cH$ separates:} By \autoref{reducetoroabp}, we see that for any $\vec{A}$ and $\vec{B}$, there is an $f\in \F[\vec{M}]^{\GL_n(\F)}$ with $f(\vec{A})\ne f(\vec{B})$ iff there is some $\ell\in[n^2]$ such that $f_\ell(\vec{A},\vec{x})-f_\ell(\vec{B},\vec{x})\ne 0$.  By \autoref{diffofinv}, $f_\ell(\vec{A},\vec{x})-f_\ell(\vec{B},\vec{x})$ is computable by a width $\le 2n^2$, depth $\ell$, degree $<r$ ROABP, which by the addition of $n^2-\ell$ dummy variables (say, to the end of $\vec{x}$), can be considered as a depth $n^2$ ROABP.  Thus, as $\cH$ is a hitting set for ROABPs of the relevant size, for any $\ell\in[n^2]$, $f_\ell(\vec{A},\vec{x})-f_\ell(\vec{B},\vec{x})\ne 0$ iff there is some $\vec{\alpha}\in\cH$ such that $f_\ell(\vec{A},\vec{\alpha})-f_\ell(\vec{B},\vec{\alpha})\ne0$, and thus iff there is an $\alpha\in\cH$ such that the invariant $f'(\vec{M})\eqdef f_\ell(\vec{M},\vec{\alpha})\in T_\cH$ separates $\vec{A}$ and $\vec{B}$.
\end{proof}

As done in Mulmuley's Theorem 3.6, we can conclude that the ring of invariants is integral over the subring generated by the separating invariants.  This uses the following theorem of Derksen and Kemper~\cite{DK} (using the ideas of geometric invariant theory~\cite{MFK}), which we only state in our specific case, but does hold more generally.

\begin{theorem}[Theorem 2.3.12, Derksen and Kemper~\cite{DK}, stated by Mulmuley in Theorem 2.11]
	\label{dkintegral}
	Let $\F$ be an algebraically closed field of characteristic zero. Let $\cT'\subseteq \F[\vec{M}]^{\GL_n(\F)}$ be a finite set of homogeneous separating invariants.  Then $\F[\vec{M}]^{\GL_n(\F)}$ is integral over the subring $S$ generated by $\cT'$.
\end{theorem}

Combining \autoref{thm:sepinv} and \autoref{dkintegral} yields the following corollary.

\begin{corollary}\label{thm:general integral}
	Assume the setup of \autoref{const:hittingset}. Let $\F$ be an algebraically closed field of characteristic zero. Then $\F[\vec{M}]^{\GL_n(\F)}$ is integral over the subring generated by $\cT_\cH$, a set of $n^2|\cH|$ invariants with $\poly(t(n,r),n,r)$-explicit ABPs.
\end{corollary}

Continuing with the dictionary of geometric invariant theory~\cite{MFK},  we can obtain the following deterministic black-box algorithm for testing of two orbit closures intersect.

\begin{corollary}
	\label{blackboxorbitclosure}
	Assume the setup of \autoref{const:hittingset}. Let $\F$ be a field of characteristic zero. There is an algorithm, running in deterministic $\poly(n,r,t(n,r),|\cH|)$-time, in the unit cost arithmetic model, such that given $\vec{A},\vec{B}\in(\F^{\zr{n}\times \zr{n}})^{\zr{r}}$, one can decide whether the orbit closures of $\vec{A}$ and $\vec{B}$ under simultaneous conjugation have an empty intersection. Further, this algorithm is black-box, as it only compares $f(\vec{A})$ and $f(\vec{B})$ for various polynomials $f$.
\end{corollary}
\begin{proof}
	As observed in Mulmuley's Theorem~3.8, the orbit closures of $\vec{A}$ and $\vec{B}$ intersect iff all invariants in $\F[\vec{M}]^{\GL_n(\F)}$ agree on $\vec{A}$ and $\vec{B}$. Our \autoref{thm:sepinv} shows this question can be answered by testing if $\vec{A}$ and $\vec{B}$ agree with respect to all $f\in\cT_\cH$, and this can be tested in $\poly(n,r,|\cH|,t(n,r))$-time, as ABPs can be evaluated quickly.
\end{proof}

Thus, the above results, \autoref{thm:sepinv}, \autoref{thm:general integral}, and \autoref{blackboxorbitclosure} give positive results to Questions~\ref{q:separating invariants}, \ref{q:integral subring}, and \ref{q:orbit intersection} respectively, assuming small explicit hitting sets for ROABPs.  Plugging in the hitting sets results of Forbes and Shpilka~\cite{ForbesShpilka12a} as cited in \autoref{thm:FS}, we obtain \autoref{thm:integral subring} and \autoref{Cor:main integral}.

However, using the hitting set of \autoref{thm:FS} does not allow us to deduce the efficient algorithm for orbit closure intersection claimed in \autoref{thm:orbit intersection} as the hitting set is too large.  To get that result, we observe that deciding the orbit closure intersection problem does not require black-box PIT, and that white-box PIT suffices.  Thus, invoking the white-box results of Raz and Shpilka~\cite{RazShpilka05}, and the follow-up work by Arvind, Joglekar and Srinivasan~\cite{ArvindJS09}, we can get the desired result.

\begin{theorem*}[\autoref{thm:orbit intersection}]
	Let $\F$ be a field of characteristic zero. There is an algorithm, running in deterministic $\polylog(n,r)$-time using $\poly(n,r)$-processors ($\NC$), in the unit cost arithmetic model, such that given $\vec{A},\vec{B}\in(\F^{\zr{n}\times \zr{n}})^{\zr{r}}$, one can decide whether the orbit closures of $\vec{A}$ and $\vec{B}$ under simultaneous conjugation have an empty intersection.
\end{theorem*}
\begin{proof}
	As observed in Mulmuley's Theorem~3.8, the orbit closures of $\vec{A}$ and $\vec{B}$ intersect iff all invariants in $\F[\vec{M}]^{\GL_n(\F)}$ agree on $\vec{A}$ and $\vec{B}$.  Our results, \autoref{reducetoroabp} and \autoref{diffofinv}, show there is a non-empty intersection iff an $n^2$-sized set of $\poly(n,r)$-size ROABPs all compute the zero polynomial.

	Raz and Shpilka~\cite{RazShpilka05} gave a polynomial-time algorithm (in the unit-cost arithmetic model) for deciding if a non-commutative ABP computes the zero polynomial, and ROABPs are a special case of the non-commutative ABP model because the oblivious nature ensures that all multiplications of the variables are in the same order and thus commutativity is never exploited.  Thus, by applying this algorithm to all of the ROABPs $f_\ell(\vec{A},\vec{x})-f_\ell(\vec{B},\vec{x})$, we can decide if the orbit closures of $\vec{A}$ and $\vec{B}$ intersect in polynomial time, thus in $\P$.

	Further, Arvind, Joglekar and Srinivasan~\cite{ArvindJS09} observed that the Raz-Shpilka~\cite{RazShpilka05} algorithm can be made parallel (within $\NC^3$) while still running in polynomial-time, by using parallel linear algebra.   Arvind, Joglekar and Srinivasan~\cite{ArvindJS09} also gave an alternate, characteristic-zero specific, parallel algorithm for this problem (within $\NC^2$). Hence, one can test each of ``$f_\ell(\vec{A},\vec{x})-f_\ell(\vec{B},\vec{x})\equiv 0$?'' in parallel, and then return the ``and'' of all of these tests, again in parallel.  Thus, this gives a parallel polynomial-time ($\NC$) algorithm for testing if orbit closures interesect.
\end{proof}

We comment briefly on space-bounded boolean computation, and its relation with this work.  As noted in Forbes and Shpilka~\cite{ForbesShpilka12a}, the ROABP model is a natural algebraic analogue of space-bounded boolean computation and the hitting sets given by Forbes and Shpilka~\cite{ForbesShpilka12a} can be seen as an algebraic analogue of the boolean pseudorandom generator (PRG) given by Nisan~\cite{Nisan92}.  First, we note that by this analogy, and the fact that subsequent work by Nisan~\cite{Nisan94} showed that the PRG of Nisan~\cite{Nisan92} can be made polynomial-time with additional space, one expects that the quasi-polynomial-time blackbox identity test (or hitting set) of Forbes and Shpilka~\cite{ForbesShpilka12a} can be made into a parallel polynomial-time whitebox identity test for ROABPs, which would bring the proof of \autoref{thm:orbit intersection} more in line with the other parts of this paper. However, the $\NC^3$ version of the Raz-Shpilka~\cite{RazShpilka05} algorithm is simpler than any modification of the Forbes-Shpilka~\cite{ForbesShpilka12a} result, we do not pursue the details here.

By this connection, it follows that one can convert the white-box PIT problem for ROABPs into a derandomization question in small-space computation (once the bit-lengths of the numbers involved are bounded).  Thus, one could avoid the algorithms of Raz-Shpilka~\cite{RazShpilka05} and Arvind, Joglekar and Srinivasan~\cite{ArvindJS09}, and simply use the algorithm of Nisan~\cite{Nisan94}.  However, this is less clean, and furthermore this booleanization cannot give similar results to \autoref{thm:sepinv}, since one cannot a priori bound the bit-length of the matrices $\vec{A}$ and $\vec{B}$.

We note here that the above result for testing intersection of orbit closures is stated in the unit-cost arithmetic model for simplicity.  At its core, the result uses linear algebra, which can be done efficiently even when the bit-lengths of the numbers are considered.  Thus, it seems likely the above algorithm is also efficient with respect to bit-lengths, but we do not pursue the details here.

\section{Equivalence of Trace of Matrix Powers, Determinants and ABPs}\label{sec:equivalent}

In this section we study the class of polynomials computed by small traces of matrix powers, to gain insight into the strength of the derandomization hypotheses that Mumuley requires for his implications regarding Noether Normalization.  In particular, we show that a polynomial can be computed as a small trace of matrix power iff it can be computed by a small ABP, as defined in \autoref{def: ABP}.

We first show that from the perspective of the hardness of derandomizing PIT, a trace of matrix power can be either defined using linear or affine forms, so that we are not losing generality in our equivalence results below, when using the affine definition.  Specifically, we want to relate the complexity of PIT for $\tr(A(\vec{x})^d)$, for an affine matrix $A(\vec{x})=A_0+\sum_{i=1}^n A_ix_i$, to the complexity of PIT of $\tr(A'(\vec{x},z)^d)$, where $A'$ is the homogenized version $A'(\vec{x},z)\eqdef A_0z+\sum_{i=1}^n A_ix_i$.  Clearly, one can reduce the homogenized linear case back to the affine case, by taking $z=1$, in both the black-box and white-box PIT model.  We now consider reducing the affine case to the linear case.  Note that this is trivial in the white-box model of PIT, as given the trace of matrix power $\tr(A(\vec{x})^d)$ we can easily construct the trace of matrix power $\tr(A'(\vec{x},z)^d)$ by replacing constants by the appropriate multiple of $z$.  The next lemma shows that these two traces are also polynomially equivalent in the black-box model.

\begin{lemma}
	\label{linervsaffine}
	Let $A(x_1,\ldots,x_n)=A_0+\sum_{i=1}^n A_ix_i$ be matrix of affine forms.  Define its homogenization $A'(\vec{x},z)=A_0z+\sum_{i=1}^n A_ix_i$.  Then for any $\vec{\alpha}$ and $\beta$, $\tr(A'(\vec{\alpha},\beta)^d)$ can be computed using $\poly(d)$ queries to $\tr(A(\vec{x})^d)$.
\end{lemma}
\begin{proof}
	\underline{$\beta\ne 0$:}  Observe that by homogeneity and linearity of the trace, we have that $\tr(A'(\vec{\alpha},\beta)^d)=\beta^d\tr(A(\vec{\alpha}/\beta)^d)$, where $\vec{\alpha}/\beta$ is the resulting of dividing $\vec{\alpha}$ by $\beta$ coordinate-wise. Thus, only 1 query is needed in this case.

	\underline{$\beta=0$:}
	Writing $y\vec{x}$ for the coordinate-wise multiplication of $y$ on $\vec{x}$, we can expand the trace of matrix power in the variable $y$, so that $\tr(A(y\vec{x})^d)=\sum_j \coeff{y^j}(\tr(A(y\vec{x})^d))$, where $\coeff{y^j}$ extracts the relevant coefficient in $y$, resulting in a polynomial in $\vec{x}$.  It follows from polynomial interpolation that in $d+1$ queries to $\tr(A(y\vec{x})^d)$ we can compute $\coeff{y^j}(\tr(A(y\vec{x})^d))$ for any $j$. Observing that $\tr(A'(\vec{\alpha},0)^d)=\coeff{y^d}(\tr(A(y\vec{x})^d))$ yields the result.
\end{proof}

Thus as $\tr(A'(\vec{x},z)^d)=0$ iff $\tr(A(\vec{x})^d)=0$, solving black-box PIT for $\tr(A'(\vec{x},z)^d)$ solves it for $\tr(A(\vec{x})^d)$, and the above lemma gives the needed query-access reduction. Thus, as the linear and affine models are equivalent with respect to PIT, we now only discuss traces of matrix powers in the affine case, and seek to show this model is computationally equivalent (not just with respect to PIT) to the ABP model. We first establish that traces of matrix powers can efficiently simulate ABPs.  To do this, we first study matrix powers and how they interact with traces.

\begin{lemma}\label{lem: power of diagonal}
	Let $x_0,\ldots,x_{d-1}$ be formally non-commuting variables, and let $R$ be any commutative ring.  Define $A(\vec{x})\in R[x_0,\ldots,x_{d-1}]^{\zr{d}\times\zr{d}}$ by
	\[A(\vec{x})_{i,j}=
		\begin{cases}
			x_i	&\text{if } j=i+1\pmod{d}\\
			0	&\text{else}
		\end{cases}
	\]
	Then $\Tr(A(\vec{x})^d)=\sum_{i\in\zr{d}}x_ix_{i+1}\cdots x_{(i-1 \bmod{d})}$.
\end{lemma}
\begin{proof}
	The matrix $A$ defines an adjacency matrix on a $d$-vertex graph, which is the directed cycle with weights $x_0,x_1,\ldots,x_{d-1}$ ordered cyclically. Raising $A$ to the $d$-power corresponds to the adjacency matrix for the length-$d$ walks on the length-$d$ cycle.  The only such walks are single traversals of the cycle, starting and ending at some vertex $i$, and these walks have weight $x_ix_{i+1}\cdots x_{(i-1\bmod{d})}$.  Taking the trace of $A^d$ corresponds to summing the weights of these walks, giving the desired formula.
\end{proof}
\begin{proof}[Alternate proof of \autoref{lem: power of diagonal}]
	By definition, \[(A^d)_{i,i}= \sum_{i_1,\ldots,i_{d-1}}A_{i,i_1} A_{i_1,i_2}\cdots A_{i_{d-1},i}.\] It follows that the only nonzero contribution is when $i_j=i_{j-1}+1$ for all $j$, when defining $i_0=i_d=i$ and working modulo $d$, and that this yields $x_ix_{i+1}\cdots x_{(i-1\bmod{d})}$. The claim follows by summing over $i$.
\end{proof}

As this result holds even when the variables $x_i$ are non-commuting, we can use this lemma over the ring of matrices and thus embed matrix multiplication (over varying matrices) to matrix powering (of a single matrix).

\begin{corollary}\label{cor: trace of power}
	Let $R$ be a commutative ring, and let $M_1,\ldots,M_d\in R^{\zr{n}\times\zr{n}}$ be matrices.  Define the larger matrix $A\in R^{\zr{nd}\times\zr{nd}}$ by treating $A$ as a block matrix in $(R^{\zr{n}\times\zr{n}})^{\zr{d}\times\zr{d}}$, so that
	\[
		A_{i,j}=
			\begin{cases}
				M_i	&\text{if } j=i+1\pmod{d}\\
				0	&\text{else}
			\end{cases}
	\]
	Then $\tr(A^d)=d\tr(M_1\cdots M_d)$.
\end{corollary}
\begin{proof}
	The properties of block-matrix multiplication imply that we can treat the $M_i$ as lying in the non-commutative ring of matrices, and thus we apply \autoref{lem: power of diagonal} to the trace of $A^d$ to see that
	\begin{align*}
		\tr(A^d)
		&=\sum_{i=1}^d\tr(M_iM_{i+1}\cdots M_{(i-1\bmod{d})})\\
		&=d\tr(M_1M_2\cdots M_{d})
	\end{align*}
	where the second equality uses that trace is cyclic.
\end{proof}

This lemma shows that the trace of a matrix power can embed the trace of a matrix product (up to the factor $d$), and \autoref{roabpmatrix} shows that traces of matrix products capture ABPs.  This leads to the equivalence of traces of matrix powers and ABPs.

\begin{theorem}\label{thm:ABP STIT}
	Let $\F$ be a field. If a polynomial $f$ is computable by a width $w$, depth $d$ ABP, then for any $d'\ge d$ such that $\chara(\F)\nmid d'$, $f$ can be computed by a width $wd'$, depth $d'$ trace of matrix power. In particular, $d'\in\{d,d+1\}$ suffices.
	
	Conversely, if a polynomial $f$ is computable by a width $w$, depth $d$ trace of matrix power, then $f$ can also be computed by a width $w^2$, depth $d$ ABP.
\end{theorem}
\begin{proof}
	\underline{ABP $\implies$ trace of matrix power:} By \autoref{roabpmatrix} there are $\zr{w}\times\zr{w}$ matrices $M_1,\ldots,M_d$ whose entries are affine forms in $\vec{x}$, such that $f(\vec{x})=\tr(M_1\cdots M_d)$.  For $i\in[d']$, define new matrices $M'_i$ such that for $i=1$,  $M'_i\eqdef M_i/d'$,  for $1<i\le d$, we set $M'_i\eqdef M_i$, and for $i>d$, define $M'_i\eqdef \Id_w$, the $\zr{w}\times\zr{w}$ identity matrix. By linearity of the trace, it follows that $f(\vec{x})=d'\cdot \tr(M_1\cdots M_d M_{d+1}\cdots M_{d'})$. Noting that the $M'_i$ have entries that are all affine forms, \autoref{cor: trace of power} implies that there is an $\zr{wd'}\times \zr{wd'}$ matrix $A(\vec{x})$ whose entries are affine forms in $\vec{x}$, such that $\Tr(A(\vec{x})^{d'})=d'\cdot \Tr(M_1\cdots M_{d'}) = f(\vec{x})$, as desired.  Noting that $d$ and $d+1$ cannot both be divisible by the characteristic of $\F$ completes the claim.

	\underline{trace of matrix power $\implies$ ABP:}  Suppose $f(\vec{x})=\tr(A(\vec{x})^d)$, where $A(\vec{x})$ is a $\zr{w}\times\zr{w}$ matrix of affine forms.  Note then that for each $i$, the $(i,i)$-th entry of $A(\vec{x})^d$ is computable by a width $w$, depth $d$ ABP, as established by \autoref{roabpmatrix}, after applying the suitable permutation of indices.  As the trace is the summation over $i$ of these functions, we can apply \autoref{addROABP} to get the result.
\end{proof}

Thus the above shows that, up to polynomial factors in size, ABPs and traces of matrix powers compute the same polynomials.  We note that there is also an equivalent computational model, defined by the determinant, which we mention for completeness.

\begin{definition}
	A polynomial $f(x_1,\ldots,x_n)$ is a width $w$ \textbf{projection of a determinant} if there exists a matrix $A(\vec{x})\in\F[x_1,\ldots,x_n]^{\zr{w}\times\zr{w}}$ whose entries are affine functions in $\vec{x}$ such that $f(\vec{x})=\det(A(\vec{x}))$. The \textbf{size} of a projection of a determinant is $nw$.
\end{definition}

Valiant~\cite{Valiant79a} (see also \cite{MalodPortier08}) showed any small ABP can be simulated by a small projection of a determinant.  Conversely, phrased in the language of this paper, Berkowitz~\cite{Berkowitz84} gave a small ABP for computing a small projection of a determinant.  Thus, the projection of determinant model is also equivalent to the ABP model.  We summarize these results in the following theorem.

\begin{theorem}\label{thm: equivalence}
	The computational models of algebraic branching programs, traces of matrix powers, and projections of determinants are equivalent, up to polynomial blow up in size.
\end{theorem}

In particular, this implies that derandomizing black-box PIT is equally hard for all three of these computational models.

\section{Hasse Derivatives}\label{sec:hasse}

\newcommand{\ply}{R[\vec{x}]}
\newcommand{\partialu}{\partial_{\vec{u}}}
\newcommand{\partialv}{\partial_{\vec{v}}}
\newcommand{\hasse}{\partial}
\newcommand{\hasseu}{\hasse_{\vec{u}^k}}
\newcommand{\hassev}{\hasse_{\vec{v}^\ell}}

In this section we define Hasse derivatives, which are a variant of (formal) partial derivatives but work better over finite fields.  For completeness, we will derive various properties of Hasse derivatives. We start with the definition.

\begin{definition}
	Let $R$ be a commutative ring, and $\ply$ be the ring of $n$-variate polynomials.  For a vector $\vec{u}\in R^n$ and $k\ge 0$, define $\hasseu(f):\ply\to\ply$, the \textbf{$k$-th Hasse derivative of $f$ in direction $\vec{u}$}, by $\hasseu(f)=\coeff{y^k}(f(\vec{x}+\vec{u}y))\in\ply$, where $\vec{x}+\vec{u}y$ is defined as the vector whose $i$-th coordinate in $x_i+u_iy$.

	If $\vec{u}=\vec{e}_i$, then we use $\hasse_{x_i^k}$ to denote $\hasse_{\vec{e}_i^k}$, and will call this \textbf{the $k$-th Hasse derivative with respect to the variable $x_i$}. For $\vec{i}\in\N^n$, we will define $\hasse_{\vec{x}^{\vec{i}}}\eqdef\hasse_{x_1^{i_1}}\cdots\hasse_{x_n^{i_n}}$.
\end{definition}

Note that for $k=1$, this is usual (formal) partial derivative. We now use this definition to establish some basic properties of the Hasse derivative.  In particular, the below commutativity property shows that the definition of $\hasse_{\vec{x}^{\vec{i}}}$ is not dependent on the order of the variables. Note that while Hasse derivatives are linear operators on $R[\vec{x}]$, they are not linear in the direction $\vec{u}$ of the derivative, and so several of these properties will be stated for more than two terms.

\begin{lemma}\label{hassebasic}
	For $f,g\in\ply$, $\vec{u},\vec{v}\in R^n$, $\alpha,\beta\in R$, and $k,\ell,k_1,\ldots,k_m\ge 0$, then
	\begin{enumerate}
		\item $\hasse_{\vec{u}^0}(f)=f$\label{hassebasic:ident}
		\item $\hasseu(\alpha f+\beta g)=\alpha\hasseu(f)+\beta\hasseu(g)$\label{hassebasic:linearf}
		\item $f(\vec{x}+\vec{u}_1y_1+\cdots+\vec{u}_my_m)=\sum_{k_1,\ldots,k_m\ge 0}\hasse_{\vec{u}_m^{k_1}}\cdots\hasse_{\vec{u}_m^{k_m}}(f)y_1^{k_1}\cdots y_m^{k_m}$\label{hassebasic:hassetaylor}
		\item $\hasse_{\vec{u}_1^{k_1}}\cdots\hasse_{\vec{u}_m^{k_m}}(f)=\coeff{y_1^{k_1}\cdots y_m^{k_m}}(f(\vec{x}+\vec{u}_1y_1+\cdots+\vec{u}_my_m))$\label{hassebasic:hassecoeff}
		\item $\hasseu\hassev(f)=\hassev\hasseu(f)$\label{hassebasic:commut}
		\item $\hasse_{\left(\sum_{j=1}^m\alpha_j\vec{u}_j\right)^k}(f)=\sum_{k_1+\cdots+k_m=k} \left(\prod_j\alpha_j^{k_j}\right) \hasse_{\vec{u}_1^{k_1}}\cdots\hasse_{\vec{u}_m^{k_m}}(f)$\label{hassebasic:directionlinearity}
		\item $\hasse_{\vec{u}^{k_1}}\cdots\hasse_{\vec{u}^{k_m}}(f)=\binom{k_1+\cdots+k_m}{k_1,\ldots,k_m}\hasse_{\vec{u}^{k_1+\cdots+k_m}}$\label{hassebasic:hassehasse}
	\end{enumerate}
\end{lemma}
\begin{proof}
	\uline{\eqref{hassebasic:ident}:} This is trivial.

	\uline{\eqref{hassebasic:linearf}:} This follows from the linearity of $\coeff{y^k}(\cdot)$, and so
	\[
		\coeff{y^k}(\alpha f(\vec{x}+\vec{u}y)+\beta g(\vec{x}+\vec{u}y))
		=
		\alpha \coeff{y^k}(f(\vec{x}+\vec{u}y))+\beta \coeff{y^k}(g(\vec{x}+\vec{u}y)))
	\]

	\uline{\eqref{hassebasic:hassetaylor}:} This will be proven by induction on $m$.

	\uline{$m=1$:} This is the definition of the Hasse derivative.

	\uline{$m>1$:}  By induction, we have that
	\[f(\vec{x}+\vec{u}_2y_2+\cdots+\vec{u}_my_m)=\sum_{k_2,\ldots,k_m\ge 0} [\hasse_{\vec{u}_2^{k_2}}\cdots\hasse_{\vec{u}_m^{k_m}}(f)](\vec{x})y_2^{k_2}\cdots y_m^{k_m}.\]
	Making the substitution $\vec{x}\leftarrow\vec{x}+\vec{u}_1y_1$ we obtain
	\[f(\vec{x}+\vec{u}_1y_1+\vec{u}_2y_2+\cdots+\vec{u}_my_m)=\sum_{k_2,\ldots,k_m\ge 0} [\hasse_{\vec{u}_2^{k_2}}\cdots\hasse_{\vec{u}_m^{k_m}}(f)](\vec{x}+\vec{u}_1y_1)y_2^{k_2}\cdots y_m^{k_m}\]
	and so by expanding $\hasse_{\vec{u}_2^{k_2}}\cdots\hasse_{\vec{u}_m^{k_m}}(f)$ into its Hasse derivatives, we obtain
	\[f(\vec{x}+\vec{u}_1y_1+\vec{u}_2y_2+\cdots+\vec{u}_my_m)=\sum_{k_1,k_2,\ldots,k_m\ge 0} [\hasse_{\vec{u}_1^{k_1}}\hasse_{\vec{u}_2^{k_2}}\cdots\hasse_{\vec{u}_m^{k_m}}(f)](\vec{x})y_1^{k_1}y_2^{k_2}\cdots y_m^{k_m}\]
	as desired.

	\uline{\eqref{hassebasic:hassecoeff}:} This is a restatement of \eqref{hassebasic:hassetaylor}.
	
	\uline{\eqref{hassebasic:commut}:} This follows immediately from \eqref{hassebasic:hassecoeff}, taking $m=2$, as $\coeff{z^\ell y^k}(f(\vec{x}+\vec{v}z+\vec{u}y))=\coeff{y^kz^\ell}(f(\vec{x}+\vec{u}y+\vec{v}z))$.

	\uline{\eqref{hassebasic:directionlinearity}:} By \eqref{hassebasic:hassetaylor} we have that
	\[f(\vec{x}+\vec{u}_1y_1+\cdots+\vec{u}_my_m)=\sum_{k_1,\ldots,k_m\ge 0} [\hasse_{\vec{u}_1^{k_1}}\cdots\hasse_{\vec{u}_m^{k_m}}(f)](\vec{x})y_1^{k_1}\cdots y_m^{k_m}\]
	and applying the substitution $y_j\leftarrow \alpha_jy$ we obtain
	\[f(\vec{x}+(\alpha_1\vec{u}_1+\cdots+\alpha_m\vec{u}_m)y)=\sum_{k_1,\ldots,k_m\ge 0} \alpha_1^{k_1}\ldots\alpha_m^{k_m}[\hasse_{\vec{u}_1^{k_1}}\cdots\hasse_{\vec{u}_m^{k_m}}(f)](\vec{x})y^{k_1+\cdots+k_m}\]
	and thus taking the coefficient of $y^k$ yields the result.

	\uline{\eqref{hassebasic:hassehasse}:} By \eqref{hassebasic:hassecoeff}, we see that
	\begin{align*}
		\hasse_{\vec{u}^{k_1}}\cdots\hasse_{\vec{u}^{k_m}}(f)=
		&=\coeff{y_1^{k_1}\cdots y_m^{k_m}}(f(\vec{x}+\vec{u}y_1+\cdots+\vec{u}y_m))\\
		&=\coeff{y_1^{k_1}\cdots y_m^{k_m}}(f(\vec{x}+\vec{u}(y_1+\cdots+y_m)))\\
		&=\coeff{y_1^{k_1}\cdots y_m^{k_m}}\left(\sum_{k}\hasseu(f)(\vec{x})\cdot(y_1+\cdots+y_m)^k\right)\\
		&=\sum_{k}\hasseu(f)(\vec{x})\cdot\coeff{y_1^{k_1}\cdots y_m^{k_m}}\left((y_1+\cdots+y_m)^k\right)\\
		&=\binom{k_1+\cdots+k_m}{k_1,\ldots,k_m}\hasse_{\vec{u}^{k_1+\cdots+k_m}}(f)
		\qedhere
	\end{align*}
\end{proof}

We now recover the action of a partial derivative on a monomial.

\begin{lemma}
	\label{hassemonomial}
	For any $\ell\in[n]$, $k\ge 0$, and $i_\ell\ge 0$,
		\[
		\hasse_{x_\ell^k}(x_1^{i_1}\cdots x_\ell^{i_\ell}\cdots x_n^{i_n})
		=
		\binom{i_\ell}{k} x_1^{i_1}\cdots x_{\ell-1}^{i_{\ell-1}}x_\ell^{i_\ell-k}x_{\ell+1}^{i_{\ell+1}}\cdots x_n^{i_n}
		\]
\end{lemma}
\begin{proof}
	We do the case where $\ell=1$, as the general case is symmetric.  Thus we want to understand the coefficient of $y^k$ in $(x_1+y)^{i_1}x_2^{i_2}\ldots x_n^{i_n}$. The binomial theorem tells us the coefficient of $y^k$ in $(x_1+y)^{i_1}$ is $\binom{i_1}{k}x_1^{i_1-k}$ (even in the case that $i_1=0$, interpreted correctly).  Plugging this into the rest of the monomial yields the result.
\end{proof}

We can now use the above properties to establish the product rule for Hasse derivatives.

\begin{lemma}[Product Rule]
	For $f,g\in\ply$, $\vec{u}\in R^n$ and $k\ge 0$, \[\hasseu(fg)=\sum_{i+j=k}\hasse_{\vec{u}^i}(f)\hasse_{\vec{u}^j}(g)\]
\end{lemma}
\begin{proof}
	\begin{align*}
		(fg)(\vec{x}+\vec{u}y)
		&=f(\vec{x}+\vec{u}y)g(\vec{x}+\vec{u}y)\\
		&=\left(\sum_i \hasse_{\vec{u}^i}(f)y^i\right) \left(\sum_j \hasse_{\vec{u}^j}(g)y^j\right)\\
		&=\sum_k\sum_{i+j=k} \hasse_{\vec{u}^i}(f)\hasse_{\vec{u}^j}(g)y^k
	\end{align*}
	and result follows by taking the coefficient of $y^k$.
\end{proof}

We now will establish the chain rule for Hasse derivatives.  As Hasse derivatives necessarily take multiple derivatives all at once, the chain rule we derive will be more complicated than the usual chain rule for (formal) partial derivatives, which was only for a single partial derivative.  This formula, and its variants, are sometimes called Fa\`{a} di Bruno's formula. The below formula is written with vector exponents, as explained in \autoref{notation}, and the $\hasse$ operators applied to vectors of polynomials are defined coordinate-wise.

\begin{lemma}[Chain Rule]
	\label{hassechainrule}
	For $f\in\ply$ and $g_1,\ldots,g_n\in R[\vec{y}]$,
	\begin{multline*}
		\hasseu(f(g_1,\ldots,g_n))
		=\sum_{\sum_{j=1}^k j|\vec{\ell}_j|=k}
			\left(\prod_{j=1}^{k}\left[\hasse_{\vec{u}^j}(\vec{g})(\vec{y})\right]^{\vec{\ell}_{j}}\right)
			\binom{\vec{\ell}_{1}+\cdots+\vec{\ell}_{k}}{\vec{\ell}_{1},\ldots,\vec{\ell}_{k}}\left[\hasse_{\vec{x}^{\sum_{j=1}^k\vec{\ell}_j}}(f)\right](\vec{g}(\vec{y}))
	\end{multline*}
\end{lemma}
\begin{proof}
	\begin{align*}
		\hasseu(f(g_1,\ldots,g_n))
		&=\coeff{z^k}\left((f\circ \vec{g})(\vec{y}+\vec{u}z)\right)\\
		&=\coeff{z^k}\left(f\left(\sum_\ell \hasse_{\vec{u}^\ell}(\vec{g})(\vec{y})z^\ell\right)\right)\\
		&=\coeff{z^k}\left(f\Big(\vec{g}(\vec{y})+\hasse_{\vec{u}}(\vec{g})(\vec{y})z+\cdots+\hasse_{\vec{u}^k}(\vec{g})(\vec{y})z^k\Big)\right)\\
	\end{align*}
	We now seek to take derivatives ``in the direction of $\hasse_{\vec{u}^j}(\vec{g})(\vec{y})$'' from the point $\vec{x}\eqdef\vec{g}(\vec{y})$, treating each $j$ as a different direction and each $z^j$ as a different variable.  However, this is a subtle operation, as up until now we have taken the directions of our derivatives as independent of the point of derivation.  To make this subtlety clear, we now study the above equation, by ``undoing'' the substitutions $\vec{x}\leftarrow \vec{g}(\vec{y})$, and $z_j\leftarrow z^j$, and working with derivatives in the ring $R[\vec{y}][\vec{x}]$.  We will then later ``redo'' these substitutions.  A simpler form of this logic was used to establish \autoref{hassebasic}.\ref{hassebasic:directionlinearity}. We start about by applying \autoref{hassebasic}.\ref{hassebasic:hassetaylor} to expand out $f$ in the directions $\hasse_{\vec{u}^j}(\vec{g})(\vec{y})$.
	\begin{align*}
		f\Big(\vec{x}+&\hasse_{\vec{u}}(\vec{g})(\vec{y})z_1+\cdots+\hasse_{\vec{u}^k}(\vec{g})(\vec{y})z_k\Big)\\
		&=\sum_{\ell_1,\ldots,\ell_k\ge 0}\left[\hasse_{\left[\hasse_{\vec{u}}(\vec{g})(\vec{y})\right]^{\ell_1}}\cdots\hasse_{\left[\hasse_{\vec{u}^k}(\vec{g})(\vec{y})\right]^{\ell_k}}(f)\right](\vec{x})\cdot z_1^{\ell_1}\cdots z_k^{\ell_k}
		\intertext{and by \autoref{hassebasic}.\ref{hassebasic:directionlinearity}, we can decompose the derivative $\hasse_{\vec{u}^j}(\vec{g})(\vec{y})$ as linear combinations of the $\hasse_{x_i^k}$ derivatives, so that as operators on $R[\vec{y}][\vec{x}]$,
		\[
			\hasse_{\left[\hasse_{\vec{u}^j}(\vec{g})(\vec{y})\right]^{\ell_j}}=\sum_{\ell_{j,1}+\cdots+\ell_{j,n}=\ell_j}\left(\prod_{i=1}^n\left[\hasse_{\vec{u}^j}(g_i)(\vec{y})\right]^{\ell_{j,i}}\right)\left[\hasse_{x_1^{\ell_{j,1}}}\cdots\hasse_{x_n^{\ell_{j,n}}}\right]
		\]
		As these operators are acting on $R[\vec{y}][\vec{x}]$ we can treat all polynomials in $\vec{y}$ as constants, in particular the terms involving the $\hasse_{\vec{u}^j}(g_i)(\vec{y})$ above.  This allows us to move all of the operators past the terms involving the $g_i$, to obtain,}
		&=\sum_{\substack{\sum_{i=1}^n \ell_{j,i}=\ell_j\\\ell_j\ge 0\\j\in[k]}}
			\left(\prod_{i=1,j=1}^{n,k}\left[\hasse_{\vec{u}^j}(g_i)(\vec{y})\right]^{\ell_{j,i}}\right)
			\left[\hasse_{x_1^{\ell_{1,1}}}\cdots\hasse_{x_n^{\ell_{1,n}}}\cdots\hasse_{x_1^{\ell_{k,1}}}\cdots\hasse_{x_n^{\ell_{k,n}}}(f)\right](\vec{x})\cdot z_1^{\ell_1}\cdots z_k^{\ell_k}
		\intertext{By invoking \autoref{hassebasic}.\ref{hassebasic:hassehasse} we can clump derivatives by variable,}
		&=\sum_{\substack{\sum_{i=1}^n \ell_{j,i}=\ell_j\\\ell_j\ge 0\\j\in[k]}}
			\left(\prod_{i=1,j=1}^{n,k}\left[\hasse_{\vec{u}^j}(g_i)(\vec{y})\right]^{\ell_{j,i}}\right)
			\left(\prod_{i=1}^n\binom{\ell_{1,i}+\cdots+\ell_{k,i}}{\ell_{1,i},\ldots,\ell_{k,i}}\right)\\
		&\hspace{2in}
			\left[\hasse_{x_1^{\sum_{j=1}^k\ell_{j,1}}}\cdots\hasse_{x_n^{\sum_{j=1}^k\ell_{j,n}}}(f)\right](\vec{x})\cdot z_1^{\ell_1}\cdots z_k^{\ell_k}
	\end{align*}
	We can ``redo'' the substitutions $\vec{x}\leftarrow \vec{g}(\vec{y})$, and $z_j\leftarrow z^j$, to obtain that,
	\begin{align*}
		\hasseu(f(g_1,\ldots,g_n))
		&=\coeff{z^k}\left(f\Big(\vec{g}(\vec{y})+\hasse_{\vec{u}}(\vec{g})(\vec{y})z+\cdots+\hasse_{\vec{u}^k}(\vec{g})(\vec{y})z^k\Big)\right)\\
		&=\sum_{\substack{\sum_{i=1}^n \ell_{j,i}=\ell_j\\\sum_{j=1}^k j\ell_j=k}}
			\left(\prod_{i=1,j=1}^{n,k}\left[\hasse_{\vec{u}^j}(g_i)(\vec{y})\right]^{\ell_{j,i}}\right)
			\left(\prod_{i=1}^n\binom{\ell_{1,i}+\cdots+\ell_{k,i}}{\ell_{1,i},\ldots,\ell_{k,i}}\right)\\
		&\hspace{2in}
			\left[\hasse_{x_1^{\sum_{j=1}^k\ell_{j,1}}}\cdots\hasse_{x_n^{\sum_{j=1}^k\ell_{j,n}}}(f)\right](\vec{g}(\vec{y}))\\
		\intertext{and rewriting things in vector notation,}
		&=\sum_{\sum_{j=1}^k j|\vec{\ell}_j|=k}
			\left(\prod_{j=1}^{k}\left[\hasse_{\vec{u}^j}(\vec{g})(\vec{y})\right]^{\vec{\ell}_{j}}\right)
			\binom{\vec{\ell}_{1}+\cdots+\vec{\ell}_{k}}{\vec{\ell}_{1},\ldots,\vec{\ell}_{k}}\left[\hasse_{\vec{x}^{\sum_{j=1}^k\vec{\ell}_j}}(f)\right](\vec{g}(\vec{y}))
			\qedhere
	\end{align*}
\end{proof}

\section{Hitting Sets for Depth-3 Diagonal Circuits}\label{sec:diagonal}

In this section we construct hitting sets for the depth-3 diagonal circuit model, as defined by Saxena~\cite{Saxena08}. In this section we will use some additional notation.  For a vector $\vec{e}\in\N^n$, we define $\supp(\vec{e})$ to be its support $S\subseteq[n]$, $|\vec{e}|_0\eqdef|\supp(\vec{e})|$ and $|\vec{e}|_\times\eqdef\prod_{\ell=1}^n(e_\ell+1)$.  We now define the depth-3 diagonal circuit model.

\begin{definition}[Saxena~\cite{Saxena08}]\label{def: diagonal}
	A polynomial $f(x_1,\ldots,x_n)$ is computable by a \textbf{depth-3 diagonal circuit} if
	\[
		f(\vec{x})=\sum_{\ell=1}^s\vec{L}_\ell^{\vec{e}_\ell}(\vec{x})\;,
	\]
	where each $\vec{L}_\ell$ is a vector of affine functions. The \textbf{size} is $n\sum_{\ell=1}^s|\vec{e}_\ell|_\times$.
\end{definition}

The hitting sets will actually be for any polynomial whose space of partial derivatives is low-dimensional.  We now define this, using the notion of Hasse derivatives from \autoref{sec:hasse} so the results apply over any characteristic.

\begin{definition}
	Let $f\in\F[x_1,\ldots,x_n]$. The \textbf{dimension of Hasse derivatives} of $f$, denoted $|\hasse(f)|$, is defined as
	\[
		|\hasse(f)|\eqdef \dim\{\hasse_{\vec{x}^{\vec{i}}}(f)|\vec{i}\in\N^n\}\;.
	\]
\end{definition}

The dimension is taken in the $\F$-vector space $\F[x_1,\ldots,x_n]$. Note that by \autoref{hassebasic} it follows that all iterated Hasse derivatives, even with arbitrary directions, are contained in the above space.  This dimension is also well-behaved in various respects, such as being sub-additive, which we now establish.

\begin{lemma}
	\label{hassedimsubadditive}
	Let $f,g\in\F[x_1,\ldots,x_n]$.  Then $|\hasse(f+g)|\le|\hasse(f)|+|\hasse(g)|$.
\end{lemma}
\begin{proof}
	As Hasse derivatives are linear (\autoref{hassebasic}) it follows that
	\[
		\spn\{\hasse_{\vec{x}^{\vec{i}}}(f+g)|\vec{i}\in\N^n\}
		\subseteq
		\spn\left(\{\hasse_{\vec{x}^{\vec{i}}}(f)|\vec{i}\in\N^n\},
			\{\hasse_{\vec{x}^{\vec{i}}}(g)|\vec{i}\in\N^n\}\right)\;,
	\]
	and thus taking dimensions finishes the claim.
\end{proof}

We now work to proving that depth-3 diagonal circuits have low-dimensional spaces of partial derivatives.  We first study the partial derivatives of a single monomial.

\begin{lemma}
	\label{hassedimmonomial}
	Let $\vec{x}^{\vec{i}}\in\F[x_1,\ldots,x_n]$.  Then $|\hasse(\vec{x}^{\vec{i}})|\le|\vec{i}|_\times$.
\end{lemma}
\begin{proof}
	Let $\vec{j}\in\N^n$, then
	\[
		\hasse_{\vec{x}^{\vec{j}}}(\vec{x}^{\vec{i}})=\prod_{\ell\in[n]}\binom{i_\ell}{j_\ell}x_\ell^{i_\ell-j_\ell}\;,
	\]
	where $\binom{i_\ell}{j_\ell}=0$ if $i_\ell<j_\ell$.  Thus, all possible Hasse derivatives are some non-zero scalar multiple of $\vec{x}^{\vec{k}}$ for any $\vec{k}\le\vec{i}$, and there are $|\vec{i}|_\times$ such $\vec{k}$, giving the desired bound.
\end{proof}

Note that this upper bound is an equality over characteristic zero, but not over finite characteristic, as seen by $x^p$ in characteristic $p$, where $|\hasse(x^p)|=2$ but $|p|_\times=p+1$. However, this slack does not qualitatively affect the results. We now extend this dimension bound by using the chain rule.

\begin{lemma}
	\label{hassedimlinearsub}
	Let $f\in\F[\vec{y}]$, and let $\vec{L}\in\F[\vec{x}]$ be a vector of affine forms.  Then $|\hasse(f\circ \vec{L})|\le|\hasse(f)|$.
\end{lemma}
\begin{proof}
	Consider some derivative $\hasse_{\vec{x}^{\vec{i}}}(f\circ\vec{L})=
	\hasse_{x_1^{i_1}}\cdots\hasse_{x_n^{i_n}}(f\circ\vec{L})$.  By the chain rule (\autoref{hassechainrule}) we see that $\hasse_{x_n^{i_n}}(f\circ\vec{L})$ is a linear combination of Hasse derivatives of $f$, evaluated at $\vec{L}(\vec{x})$, as the Hasse derivatives of $\vec{L}(\vec{x})$ are constants. Since \autoref{hassebasic} shows that derivatives of derivatives are (scalar multiples of) derivatives, we see that we can induct downwards on $\ell$ to obtain that $\hasse_{x_\ell^{i_\ell}}\cdots\hasse_{x_n^{i_n}}(f\circ\vec{L})$ is a linear combination of Hasse derivatives of $f$, evaluated at $\vec{L}(\vec{x})$.  Taking $\ell=1$, and noting that evaluating the Hasse derivatives of $f$ at $\vec{L}(\vec{x})$ cannot increase dimension, the claim follows.
\end{proof}

Combing the above result with sub-additivity of the dimension of derivatives, we can now bound the dimension of depth-3 diagonal circuits.  Such bounds are not attainable for the depth-4 diagonal circuit model (which we did not define), as that model contains the polynomial $(\sum_{\ell=1}^n x_\ell^2)^n$, which has an exponentially large space of derivatives.  As mentioned in the introduction, there are other works (\cite{AgrawalSS12} and \cite{ForbesShpilka12a}) that handle the depth-4 case, using other techniques.

\begin{lemma}
	\label{hassedimdiagonal}
	Let $f(\vec{x})=\sum_{\ell=1}^s\vec{L}_\ell(\vec{x})^{\vec{e}_\ell}$ be a depth-3 diagonal circuit, where $\vec{L}_\ell(\vec{x})$ is an affine function. Then $|\hasse(f)|\le \sum_{\ell=1}^s |\vec{e}_\ell|_\times$.
\end{lemma}
\begin{proof}
	By sub-additivity of dimension of Hasse derivatives (\autoref{hassedimsubadditive}) it suffices to prove the claim for $s=1$.  Thus consider some $f(\vec{x})=\vec{L}(\vec{x})^{\vec{e}}$.  Note that $f(\vec{x})=g(\vec{L}(\vec{x}))$, where $g(\vec{y})=\vec{y}^{\vec{e}}$.  The claim then follows from \autoref{hassedimlinearsub} and \autoref{hassedimmonomial}.
\end{proof}

We now seek to show that any polynomial with low-dimensional derivatives must have a small-support monomial.  To do so, we introduce the notion of a monomial ordering (see \cite{CLO} for more on monomial orderings) and establish facts about its interactions with derivatives.

\begin{definition}
	A \textbf{monomial ordering} is a total order $\prec$ on the non-zero monomials in $\F[\vec{x}]$ such that
	\begin{itemize}
		\item For all $\vec{i}\in\N^n$, $1\prec \vec{x}^{\vec{i}}$.
		\item For all $\vec{i},\vec{j},\vec{k}\in\N^n$, $\vec{x}^{\vec{i}}\prec\vec{x}^{\vec{j}}$ implies $\vec{x}^{\vec{i}+\vec{k}}\prec\vec{x}^{\vec{j}+\vec{k}}$
	\end{itemize}
\end{definition}

For concreteness, one can consider the lexicographic ordering on monomials, which is easily seen to be a monomial ordering.  We now observe that derivatives are monotonic with respect to monomial orderings, except when the characteristic prevents it.  That is, over characteristic $p$ we have $x^{p-1}\prec x^p$ (in any ordering), but $\hasse_x(x^p)=0$, which is not included in the ordering.

\begin{lemma}
	\label{hasseordmonotone}
	Let $\prec$ be a monomial ordering on $\F[\vec{x}]$. Let $\vec{x}^{\vec{i}}\prec\vec{x}^{\vec{j}}$ be monomials.  Then for any $\vec{k}$, if
	$\hasse_{\vec{x}^{\vec{k}}}(\vec{x}^{\vec{i}}),\hasse_{\vec{x}^{\vec{k}}}(\vec{x}^{\vec{j}})\ne 0$ then $\hasse_{\vec{x}^{\vec{k}}}(\vec{x}^{\vec{i}}) \prec \hasse_{\vec{x}^{\vec{k}}}(\vec{x}^{\vec{j}})$, where we abuse notation and ignore (non-zero) coefficients in the ordering.
\end{lemma}
\begin{proof}
	By \autoref{hassemonomial}, the assumptions of $\hasse_{\vec{x}^{\vec{k}}}(\vec{x}^{\vec{i}}),\hasse_{\vec{x}^{\vec{k}}}(\vec{x}^{\vec{j}})\ne 0$ imply that $\vec{k}\le\vec{i},\vec{j}$, and that
	\begin{align*}
		\hasse_{\vec{x}^{\vec{k}}}(\vec{x}^{\vec{i}})
		&=a\vec{x}^{\vec{i}-\vec{k}},
		&\hasse_{\vec{x}^{\vec{k}}}(\vec{x}^{\vec{j}})
		&=b\vec{x}^{\vec{j}-\vec{k}},
	\end{align*}
	where $a,b\ne 0$ are constants.  As the monomial ordering is total, and is monotonic over multiplication, it follows $\vec{x}^{\vec{i}-\vec{k}}\prec\vec{x}^{\vec{j}-\vec{k}}$, as desired.
\end{proof}

Monotonicity then implies that the largest monomial of a polynomial $f$ will continue to be largest when taking derivatives, as long as it is not annihilated.  Treating polynomial as vectors, this then gives us a diagonal system of vectors, from which we can deduce the following rank bound.

\begin{theorem}
	\label{smallmonomial}
	Let $f\in\F[\vec{x}]$ be a polynomial, and let $\prec$ be any monomial ordering in $\F[\vec{x}]$. Let $\vec{x}^{\vec{i}}$ be the largest monomial (with respect to $\prec$) with a non-zero coefficient in $f$.  Then $|\vec{i}|_0\le\log|(\hasse(f)|$.
\end{theorem}
\begin{proof}
	Consider the set of vectors $A\subseteq \N^n$ defined by
	\[
		A\eqdef\{\vec{j}: \forall \ell\in[n], j_\ell\in\{0,i_\ell\}\}\;,
	\]
	so that all vectors in $A$ have support contained in $\supp(\vec{i})$.	For a fixed $\vec{j}\in A$, linearity and \autoref{hassemonomial} imply that
	\[
		\hasse_{\vec{x}^{\vec{j}}}(\vec{x}^{\vec{i}})=\prod_{\ell\in \supp(\vec{i})\setminus\supp(\vec{j})}x_\ell^{i_\ell}\;,
	\]
	which in particular is non-zero. Write $f$ as $f=a\vec{x}^{\vec{i}}+g$, where all monomials in $g$ are less than $\vec{x}^{\vec{i}}$ and $a\ne 0$.  By \autoref{hasseordmonotone} it follows that all non-zero monomials in $\hasse_{\vec{x}^{\vec{j}}}(g)$ are less than $\hasse_{\vec{x}^{\vec{j}}}(\vec{x}^{\vec{i}})$, so $a\prod_{\ell\in \supp(\vec{i})\setminus\supp(\vec{j})}x_\ell^{i_\ell}$ is the leading term of $\hasse_{\vec{x}^{\vec{j}}}(f)$. Thus, ranging $\vec{j}$ over all vectors in $A$, we get $2^{|\vec{i}|_0}$ derivatives of $f$, each with a different leading monomial.  Thus, these polynomials are linearly independent, so it must be that $2^{|\vec{i}|_0}\le|\hasse(f)|$, giving the desired bound.
\end{proof}

Just as in \autoref{hassedimmonomial}, the characteristic of the underlying field affects the tightness of this result.  In particular, in characteristic zero, one can improve this result to $|\vec{i}|_\times\le\log|\hasse(f)|$.

The above lemma shows that any $f$ with a small-dimensional space of derivatives must have a small-support monomial.  We now give a construction aimed at hitting any polynomial with such small-support monomials.  Note that this will beat the union bound, in the sense that a union-bound argument for creating hitting sets against polynomials with small-support monomials will not yield small hitting sets as there are too many such polynomials.  However, there is still a small hitting set, as we now construct.

\begin{construction}
	\label{hitmonomial}
	Let $n,d,m\ge 1$. Let $\F$ be a field of size $\ge d+1$. Let $S\subseteq\F$ with $|S|=d+1$. Define $\cH'\subseteq\F^n$ by
	\[
		\cH'\eqdef\{\vec{\alpha}:\vec{\alpha}\in S^n, |\vec{\alpha}|_0\le m\}\;.
	\]
\end{construction}

We now establish the desired properties of this construction.

\begin{theorem}
	\label{hitmonomialworks}
	Assume the setup of \autoref{hitmonomial}. Then $\cH'$ is $\poly(n,d,m)$-explicit and has size $|\cH'|\le (nd)^m$, and for any $f\in\F[x_1,\ldots,x_n]$ of total degree $\le d$, with $|\hasse(f)|\le 2^m$, $f=0$ iff $f|_{\cH'}\equiv 0$.
\end{theorem}
\begin{proof}
	\uline{$\cH'$ is explicit:} This is clear from construction.

	\uline{$|\cH'|\le (nd)^m$:} This is clear from construction.

	\uline{$f=0 \implies f|_{\cH'}\equiv 0$:} This is clear.

	\uline{$f\ne 0 \implies f|_{\cH'}\not\equiv 0$:} By \autoref{smallmonomial} it follows that $f$ has a non-zero coefficient on a monomial $\vec{x}^{\vec{i}}$ where $\vec{i}$ has support $T\subseteq[n]$, with $|T|\le\log|\hasse(f)|\le m$. Let $\vec{y}$ be a vector of variables indexed by $T$, and define $f_T(\vec{y})$ as $f(\vec{x})$ under the substitution that $x_\ell\leftarrow 0$ for $\ell\notin T$ and $x_\ell\leftarrow y_\ell$ for $\ell\in T$. As $\vec{x}^{\vec{i}}$ is not annihilated by this substitution, it follows that $f_T\ne0$ and is of total degree $\le d$.  By construction $\cH'|_T$ contains all tuples in $S^{|T|}$.  As $|S|\ge d+1$ it follows from polynomial interpolation that $(f_T)|_{(\cH'|_T)}\not\equiv0$, and thus $f$ has a non-zero evaluation $\vec{\alpha}\in S^n$ with $\supp(\alpha)\subseteq T$.  By construction of $\cH'$, $\vec{\alpha}\in\cH'$, and thus $f|_{\cH'}\not\equiv0$ as desired.
\end{proof}

By combining \autoref{hitmonomialworks} with \autoref{hassedimdiagonal} we obtain the following hitting set for diagonal circuits.

\begin{corollary}
	\label{diagonalpit}
	Let $\F$ be a field with size $\ge d+1$. Then there is a $\poly(n,d,\log(s))$-explicit hitting set of size $\poly(n,d)^{\O(\log s)}$ for the class of $n$-variate, degree $\le d$, depth-3 diagonal circuits of size $\le s$.
\end{corollary}

\section{Acknowledgments}

The first author would like to thank Peter B\"{u}rgisser for explaining the work of Mulmuley, and the Complexity Theory week at Mathematisches Forschungsinstitut Oberwolfach for facilitating that conversation. He would also like to thank Sergey Yekhanin for the conversation that led to \autoref{orbitmembership}, and Scott Aaronson for some helpful comments.

The authors are grateful to Josh Grochow~\cite{Grochow13} for bringing the works of Chistov-Ivanyos-Karpinski~\cite{ChistovIK97}, Sergeichuk~\cite{Sergeichuk2000} and Belitski{\u \i}~\cite{belitskii1983} to their attention, and for explaining these works to them. 

\bibliographystyle{alphaurl}
\bibliography{bibliography}

\appendix

\section{Orbit Intersection Reduces to PIT}\label{sec:orbitmem}

In this section, we study the (non-closed) orbit intersection problem, as compared with the orbit closure intersection problem studied in \autoref{sec:main}.  Unlike with orbit closures, the orbits with non-empty intersections must be equal, because the group action is invertible.  Thus, the orbit intersection problem is equivalent to the orbit membership problem. As mentioned before, Chistov, Ivanyos, and  Karpinski~\cite{ChistovIK97} observed a randomized algorithm for the orbit membership problem, based on the Schwartz-Zippel lemma \cite{Schwartz80,Zippel79}. In conversations with Yekhanin~\cite{Yekhanin2013}, we also discovered this result, which we include for completeness, given its closeness to the other questions studied in this work.

\begin{theorem}
	\label{orbitmembership}
	Let $\F$ be a field of size $>n$. There is a reduction, running in deterministic $\polylog(n,r)$-time using $\poly(n,r)$-processors in the unit cost arithmetic model, from the orbit membership problem of $(\F^{\zr{n}\times \zr{n}})^{\zr{r}}$ under simultaneous conjugation, to polynomial identity testing of ABPs.  In particular, the orbit membership problem can be solved in randomized $\polylog(n,r)$-time with $\poly(n,r)$-processors ($\RNC$), in the unit cost arithmetic model.
\end{theorem}
\begin{proof}
	Let $\vec{A},\vec{B}\in(\F^{\zr{n}\times\zr{n}})^{\zr{r}}$. There exists an invertible $P$ such that $\vec{B}=P\vec{A}P^{-1}$ iff there is an invertible $P$ such that $\vec{B}P=P\vec{A}$.  This second equation is a homogeneous linear equation in $P$, and thus the set of solutions is a vector space. Gaussian elimination can efficiently (in parallel, see Borodin, von zur Gathen and Hopcroft~\cite{BorodinGH82}, and the derandomization by Mulmuley~\cite{Mulmuley87}) find a basis $\{P_i\}_{i=1}^\ell$ for this vector space, with $\ell\le n^2$.  It follows then that such an invertible $P$ exists iff $\{P_i\}_i$ contain an invertible matrix in their $\F$-span.  As the non-vanishing of the determinant characterizes invertible matrices, it follows that such a $P$ exists iff $f(x_1,\ldots,x_\ell)\eqdef\det(\sum_{i=1}^\ell P_i x_i)$ has a non-zero point in $\F^\ell$.

	Clearly $f(\vec{x})$ having a non-zero point in $\F^\ell$ implies that $f(\vec{x})$ is non-zero as a polynomial. Conversely, as the total degree of $f(\vec{x})$ is $\le n$, and the field $\F$ has size $>n$, polynomial interpolation implies that if $f(\vec{x})$ is a non-zero polynomial then it has a non-zero points in $\F^\ell$. Thus, we see that there is a $P$ such that $\vec{B}=P\vec{A}P^{-1}$ iff $f(\vec{x})$ is a non-zero polynomial, where by the results of Berkowitz~\cite{Berkowitz84}, $f(\vec{x})$ is computable by a $\poly(n,r)$-size ABP.  Thus, we have reduced orbit-membership to PIT of ABPs, and done so in parallel.  Using that ABPs can be evaluated efficiently in parallel, and that we can solve PIT with random evaluations by the Schwartz-Zippel lemma, we get the corresponding randomized parallel algorithm for orbit membership.
\end{proof}
\end{document}